\documentclass{article}
\usepackage{lscape}
\usepackage{graphicx} 
\usepackage{amsfonts}
\usepackage{amssymb}  
\usepackage{calrsfs}
\usepackage{mathrsfs}  
\usepackage{amsmath}
\usepackage{url}
\usepackage{amsthm}
\usepackage[a4paper, total={6in, 8in}]{geometry}
\usepackage{subcaption} 
\usepackage{float}    

\usepackage{tikz}
\usetikzlibrary{arrows.meta} 
\usetikzlibrary{positioning, arrows.meta} 
\usetikzlibrary{calc}
\usepackage{bm}  
\usepackage{subcaption}

\newtheorem{remark}{Remark} 
\newtheorem{definition}{Definition}
\newtheorem{assumption}{Assumption}
\newtheorem{conjecture}{Conjecture}

\allowdisplaybreaks

\usetikzlibrary{positioning}
\usetikzlibrary{patterns}
\usetikzlibrary{fit,shapes.geometric}
\newtheorem{theorem}{Theorem}  
\newtheorem{lemma}{Lemma}  
\newtheorem{proposition}{Proposition}  
\newtheorem{corollary}{Corollary}  

\newtheorem*{restatement_lem}{\normalfont\bfseries Lemma}
\newtheorem*{restatement_cor}{\normalfont\bfseries Corollary}

\usepackage{xcolor}

\title{Asymptotic behavior for a general class of spreading models}
\author{K.M.D.\ Chan, D.T.\ Crommelin, and M.R.H.\ Mandjes}
\date{\today}

\begin{document}

\maketitle

\begin{abstract}
\noindent Growing literatures on epidemic and rumor dynamics show that infection and information coevolve. We present a unified framework for modeling the spread of infection and information: a general class of interaction-driven fluid-limit models expressed as coupled ODEs. The class includes the SIR epidemic model, the Daley-Kendall rumor model, and many extensions. For this general class, we derive theoretical results: under explicit graph-theoretic conditions, we obtain a classification of asymptotic behavior and motivate a conjecture of exponential decay for vanishing states. When these conditions are violated, the classification can fail, and decay may become non-exponential (e.g., algebraic). In deriving the main result, we establish asymptotic stability and $L^1$-integrability properties for state variables. Alongside these results, we introduce the \textit{dependency graph} that captures outflow dependencies and offers a new angle on the structure of this model class. Finally, we illustrate the results with several examples, including a heterogeneous rumor model and a rumor-dependent SIR model, showing how small changes to the dependency graph can flip asymptotic behavior and reshape epidemic trajectories.
\begin{quotation}

\bigskip


\smallskip

    \noindent K.M.D.\ Chan is with Transtrend, Rotterdam, the Netherlands. He is also affiliated with Mathematical Institute, Leiden University, P.O. Box 9512, 2300 RA Leiden, The Netherlands.\url{dannychan62@gmail.com}

    \noindent D.T.\ Crommelin is with Centrum Wiskunde \& Informatica (CWI), Amsterdam, the Netherlands. He is also affiliated with Korteweg-de Vries Institute for Mathematics, University of Amsterdam, Amsterdam, The Netherlands. \url{daan.crommelin@cwi.nl}

    \noindent M.R.H.\ Mandjes is with Mathematical Institute, Leiden University, P.O. Box 9512,
2300 RA Leiden,
The Netherlands. He is also affiliated with Korteweg-de Vries Institute for Mathematics, University of Amsterdam, Amsterdam, The Netherlands; E{\sc urandom}, Eindhoven University of Technology, Eindhoven, The Netherlands; Amsterdam Business School, University of Amsterdam, Amsterdam, The Netherlands. \url{m.r.h.mandjes@math.leidenuniv.nl}

\smallskip

\noindent
{\it Acknowledgments.} 
MM's research has been funded by the NWO Gravitation project N{\sc etworks}, grant number 024.002.003.
\end{quotation}
    
\end{abstract}

\newpage 

\section{Introduction}

Contagion plays a central role in modern societies. Pathogens spread through interpersonal networks, while rumours, opinions, and misinformation cascade across social and mass media — shaping behavior and, in turn, influencing outcomes such as market dynamics and election results.
The spread of pathogens and the spread of rumours are closely intertwined: outbreaks can spark misinformation, while rumours — such as anti-vaccination narratives — can lower vaccination rates and thereby increase transmission.
The COVID-19 crisis highlighted this coupled dynamic, combining a biological pandemic with an `infodemic' of misinformation \cite{WHO2020Infodemic}. In this paper, we develop a unified modeling framework for interacting processes, such as those at play in the spread of pathogens and rumors. We analyze the asymptotic behavior of models in this unified class.

\subsubsection*{Rumours, pathogens and their interdependence}

What do we mean by {\it rumours}? In this paper, we group the spread of intangible content such as rumours, opinions, and (mis)information under the umbrella label \textit{rumour}. In the sociological literature, a rumour is an unverified --- often plausible --- piece of information that spreads widely, typically fueled by uncertainty or anxiety. Transmission occurs via word of mouth and, increasingly, social media; see the survey~\cite{Turenne2018}. Rumour dynamics can take different forms. For example, Crane~\cite{Crane2008} empirically studied rumour diffusion on platforms such as YouTube, showing that some content goes viral while other content spreads more slowly, depending, among other factors, on the content’s characteristics and the network through which it propagates. Another striking example is that false news spreads faster and further than true news~\cite{Vosoughi2018}. These dynamics have consequences across many domains of society. For example, rumours and related false news influence economic behavior --- shaping market dynamics and investment decisions \cite{Kosfeld2005,Bannerjee1993,Follmer1974} --- and, in extreme cases, can lead to serious harm: Bhavnani et al. \cite{Bhavnani2009} document instances in which rumours triggered communal violence. 

Turning to biological contagion, the study of pathogen propagation (epidemiology) is generally regarded as more established and prominent than the study of rumours — a contrast that became especially apparent during the society-disrupting COVID-19 pandemic. So, what is a pathogen? A pathogen is typically a biological agent --- most commonly a virus, bacterium, or parasite --- that infects hosts (e.g., humans, bats, or swine) and spreads via routes such as respiratory aerosols and direct contact \cite{Heymann2015,Mandell2020}. Such pathogens can vary widely in infectiousness, incubation period, and the prevalence of asymptomatic transmission \cite{Gandhi2020,Oran2020}. Their societal impact is substantial: as COVID-19 showed, successive waves and variants strained health-care systems, while interventions disrupted schooling, work, and supply chains \cite{WHO2021COVID,Nalbandian2021}. Beyond COVID-19, pathogens such as SARS (2002–2004), HIV/AIDS, dengue, and seasonal influenza have imposed comparable societal burdens \cite{Lipsitch2003,UNAIDS2021,Guzman2015,Taubenberger2006}.

These two forms of contagion are not independent but closely interlinked. Epidemiological outbreaks can trigger surges of rumours --- for example, about masking, distancing, and vaccination --- which in turn dampen or amplify transmission by shifting protective behaviors \cite{Teslya2022,Lin2020,Centola2010}. Conversely, rumours alleging vaccine harms have reduced measles vaccination coverage, leading to sizeable outbreaks --- for example, in South Wales (2012–2013) \cite{Jansen2003,Wise2013,PHWales2013}. Taken together, these examples highlight the importance of studying the processes jointly rather than in isolation.

\subsubsection*{Mathematical models of the spread of pathogens and rumours}
Mathematical modelling has been central to understanding how pathogens and rumours disseminate. Each field rests on a widely accepted, canonical baseline model: the SIR model for infectious diseases and the Daley-Kendall (DK) model for rumour propagation. In SIR, a population is partitioned into Susceptible, Infectious, and Recovered compartments; transmission arises from contact between susceptible and infectious individuals and recovery leads to immunity \cite{Kermack1927}. By contrast, while sharing a similar transmission mechanism, the DK model, tracks Ignorants (unaware), Spreaders (actively transmitting), and Stiflers (no longer transmitting); here, spreaders convert ignorants, whereas interactions among spreaders/stiflers and spreaders/spreaders extinguish further spreading \cite{DaleyKendall1965}. These baseline models have served as foundations for extensive theoretical development, with numerous extensions incorporating host heterogeneity, explicit contact-network structures, and time-varying parameters; see, for example, \cite{Diekmann1990,Bootsma2024,Barrat2008,PastorSatorrasReview2015,Nekovee2007,Zhu2020,Zhu2021}.

Given the wide range of existing model variants, a natural question arises: can these canonical models — and their extensions — be captured within a single, general mathematical framework? To the best of our knowledge, no such unifying approach has been established before. A unified perspective offers a higher-level view of both rumour and pathogen dynamics. It reveals structural insights — such as shared asymptotic behaviors — that often remain hidden when individual models are studied in isolation \cite{Teslya2022,Lin2020}.
One promising route toward unification lies in identifying what the models have in common. At their core, both SIR and DK are driven by interaction-based transitions: state changes occur through pairwise encounters. In each case, who meets whom — and in which states — determines what happens next. A susceptible individual becomes infected upon contact with someone infectious; an ignorant becomes a spreader after interacting with another spreader. Thus, the interaction process drives the transition process. Contact patterns determine which transitions are possible and how frequently they occur. This interaction-based structure forms a common backbone for SIR, DK, and many of their extensions.

This shared structure suggests a natural next step: to formulate a general, interaction-driven modelling framework that captures the essential dynamics of both SIR- and DK-type systems. In the \textit{fluid limit} --- where the population is assumed to be large enough that demographic stochasticity can be neglected~\cite{Kurtz1981} --- the resulting dynamics are deterministic and can be described by a system of coupled differential equations.
More precisely, the interaction mechanism naturally gives rise to \textit{bilinear (quadratic)} terms, leading to systems of the form:
\[
\dot y_i \;=\; \sum_{j,k}\,\alpha^{i}_{jk}\,y_j y_k,
\]
where \( y_i \equiv y_i(t) \) denotes the (normalized) fraction of individuals in state~\( i \) (e.g., susceptible or ignorant), and the coefficients \( \alpha^{i}_{jk} \in \mathbb{R} \) capture the in- and outflow rates due to pairwise encounters (the formulation also admits linear terms --- e.g., natural recovery --- under a constant-population assumption).
This mathematical structure is in the spirit of several classical frameworks, such as the \textit{Lotka--Volterra equations} for interacting species~\cite{Lotka1925,Volterra1926,Murray2002}, \textit{evolutionary dynamics}~\cite{HofbauerSigmund1998,Nowak2006,Sandholm2010}, and \textit{mass-action chemical reaction networks}~\cite{HornJackson1972,Feinberg2019}, where bimolecular interactions also yield quadratic terms.
The analogy should be considered as structural rather than literal: the models describe different entities and processes, but the shared mathematical form offers both conceptual insights (e.g., around persistence or extinction) and useful analytical tools.

\subsubsection*{Goal \& contributions}
Building on the observed empirical interdependence between pathogen and rumour contagion, and leveraging the advantages of the general interaction-driven framework outlined above, we pursue two main goals.  
(i) We introduce a unified, highly general class of contact-driven models that encompasses SIR, DK, and many of their extensions within a single formulation.  
(ii) We establish structural results for this class of models. Specifically, our main theorem provides a classification criterion: under specific (acyclicity-type) conditions, the underlying graph-structure alone determines, for each state, whether it goes extinct or persists in the long run --- that is, whether \(\lim_{t \to \infty} y_i(t)\) is positive or zero, where, as before, \(y_i(t)\) denotes the fraction of the population in state \(i\) at time \(t\).

In developing the classification rule, we proceed in stages. First, under suitable graph-theoretic conditions, we show that trajectories converge to equilibrium, in that \(\dot y_i(t)\to 0\) and \(y_i(t) \to c_i\), as \(t\to\infty\), for every state \(i\), with \(c_i\) a constant. Next, imposing additional acyclicity-type conditions on the interaction graph, we prove that any vanishing state --- i.e., one with \(\lim_{t\to\infty} y_i(t)=0\) --- is \(L^1\)-integrable, \(\int_0^\infty y_i(t)\,dt<\infty\). This integrability yields our main result: a criterion for persistence versus extinction that depends only on graph structure. The same analysis motivates a conjecture that, under these conditions, every vanishing state decays exponentially. Finally, we show that when the requisite graph-theoretic conditions are not met, the classification need not hold, and vanishing states may decay non-exponentially (e.g., algebraically).

To formalize the graph-theoretic conditions described above, we introduce a novel concept: the {\it dependency graph}. This graph captures dependencies between states by indicating which states influence others’ outflow rates. This representation forms the basis for the acyclicity-type conditions used in the classification rule.

Through a series of examples, we illustrate how the theory applies in practice. We begin by showing that if the necessary graph-theoretic conditions are not met, both exponential decay and the persistence/vanishing classification can fail. Next, in order to exemplify how models in existing literature fit within our framework, we apply it to the rumour model of Chen \& Wang (2020) \cite{Chen2020}. Our results immediately imply that the Ignorants group cannot, by design, be fully converted into spreaders, and that the dynamics converge to equilibrium: for each state \(i\), \(\dot y_i(t)\to 0\) and \(y_i(t)\to c_i\) as \(t\to\infty\).
 We then examine a heterogeneous rumour model to show how small changes in graph structure can produce large shifts in long-term behavior, offering insight into how specific characteristics of individual-level spreading behavior alter macro-level outcomes. Finally, we consider a hybrid model combining the SIR and DK frameworks, uncovering intuitively surprising relationships that do not occur in the standard SIR setting.

\subsubsection*{Organization of paper}
The remainder of the paper is organized as follows. Section~\ref{sec:model} introduces the model and the notion of a dependency graph. Section~\ref{sec:results} presents our theoretical results. Section~\ref{sec:examples} provides examples that illustrate these results in practice. Finally, Section~\ref{sec:discussion} highlights key considerations and outlines directions for future work.

\section{Model}\label{sec:model}
In this section we first provide a set of definitions that play a pivotal role in our model and analysis. We then proceed by detailing our interaction-driven dispersion model.

\subsection{A few graph-theoretic definitions}
We begin by recalling a few graph-theoretic concepts, primarily for completeness, as these notions are frequently employed in the literature. Additionally, the definitions provided will introduce some of the notation used throughout this paper.

\begin{definition}
    A \textbf{directed graph} $(V,E)$ is a set of $N$ vertices $V = {v_i: i=1, \ldots,N}$ and a set of directed edges $E \subseteq ({(v_k,v_l)}: v_k, v_l \in V)$, where ${(v_k, v_l)} \in E$ if and only if there is a directed edge from $v_k$ to $v_l$, which we denote by $v_k \to v_l$.
\end{definition}

\begin{definition}
    A \textbf{subgraph} of a directed graph $(V,E)$ is a directed graph $(U,K)$ such that $U \subseteq V$ and $K \subseteq E$. 
\end{definition}

\begin{definition} \label{def:sink_source}
In a directed graph $(V,E)$, for an edge $(v_k, v_l) \in E$, $v_k$ is called a \textbf{predecessor} of $v_l$, and $v_l$ is called a \textbf{successor} of $v_k$. A vertex is called a \textbf{sink} if it only has predecessors. A vertex is a called a \textbf{source} if it only has successors.
\end{definition}
\begin{definition}\label{def:cycle}
    In a directed graph $(V,E)$, a \textbf{path} exists from vertex $u$ to vertex $v$, if there exists a set of vertices $k_1, \ldots k_n$ such that $u \to k_1 \to k_2 \to \ldots k_n \to v$. A \textbf{cycle} is a path $u \to k_1 \to k_2 \to \ldots k_n \to v$ where $u = v$. A cycle $u \to u$ is called a \textbf{self-loop}.
\end{definition}

\begin{definition}\label{def:dag}
    A directed graph is a \textbf{Directed Acyclic Graph} (DAG) if it has no cycles. 
\end{definition}

\subsection{Interaction-driven dispersion model}

In this subsection, we introduce our interaction-driven dispersion model, where \textit{dispersion} here captures the spread caused by interactions. We will later show that this model encompasses well-known frameworks such as the SIR epidemiological model, rumor-spreading models like the DK model, and many of their extensions. Importantly, all theoretical results derived within our general framework apply directly to these specific cases. Because the model unifies both epidemiological and rumor dynamics, it serves as a powerful tool for studying the {\it interaction} between these two distinct processes (see Section \ref{subsect:sir_dk_hybrid} for an example).

\subsubsection{Rationale behind the model}

Our objective is to develop a foundational fluid-limit framework that allows for a deeper examination of various relevant phenomena, within both mathematical epidemiology and rumor dynamics. Instead of analyzing specific models in isolation, we aim to uncover structural properties of this general class of fluid-limit models under various conditions.

For tractability, we adopt two common simplifying assumptions. First, we assume that the observation time window is sufficiently short to rule out significant changes in population size, effectively imposing a conservation law. Second, we assume uniform mixing within the population.

To ensure the framework's relevance in both mathematical epidemiology and rumor dynamics, we require the following essential features:
\begin{itemize}
    \item [$\circ$]\textit{Interaction-driven state transitions:} changes in the state of individuals within the system are driven solely by their interactions.
    \item [$\circ$]\textit{Natural outflow mechanism:} the model incorporates mechanisms for individuals to exit their current state over time, such as recovery or waning susceptibility.
\end{itemize}
As the subsequent subsections demonstrate, a pure interaction-driven dispersion model satisfies these requirements.

\subsubsection{States, conservation law and initial conditions}
We throughout assume a system comprising of $N\in{\mathbb N}$ states, and we denote the fraction of agents in each state $i=1,\ldots,N$ at time $t\ge 0$ by $y_i(t)\in(0,1)$. We assume a conservation law on the population size: $\sum_{1 \leq i \leq N} y_i(t) = 1$ for any $t$. Initial conditions are assumed to be $y_i(0)\in(0,1)$. 

\subsubsection{Transition mechanism described by differential equations}\label{sec:trans_mech_de}
In this subsection we set up a system of coupled differential equations that describe our interaction-driven dispersion model. 
The probability of any agent interacting with an agent in state $i$ is equal to $y_i$ due to the assumption of uniform mixing. Let $\beta^i_{ij} \geq 0$ capture both the (homogeneous) contact rate and the (state-dependent) probability of a state $i$-agent changing state due to interaction with a state-$j$ agent (where it is stressed that the order of the indices matters). For a state $i$-agent to leave state $i$ it must interact with some other agent. Hence, the outflow of state $i$, written as $O_i(t)$, is given by the following expression:
\begin{equation}\label{eq:outflow}
    O_i(t) := y_i \sum_{1 \leq m \leq N} \beta^i_{im} y_m.
\end{equation} 
Due to the conservation law, which stipulates an unchanging total population size, each component $\beta^i_{im} y_i y_m$ representing an outflow from state $i$ must correspond to an equivalent inflow into states other than $i$. We define $\beta^k_{im} \geq 0$ (where $k \neq i$) to represent both the homogeneous contact rate and the state-dependent probability of an agent in state $i$ transitioning to state $k$ as a result of interaction with an agent in state $m$. The conservation law is therefore expressed as
\begin{equation}\label{eq:cons_law}
    \beta^i_{im} = \sum_{k \ne i} \beta^k_{im}.
\end{equation}
The inflow into state $i$ consists of all terms originating from interactions that cause agents to transition into state $i$ from any other state. This inflow is therefore defined as
\begin{equation}\label{eq:inflow}
I_i(t) :=  \sum_{\substack{1 \leq n,k \leq N, \\n \ne i}} \beta^i_{nk} y_n y_k.
\end{equation}
Hence, we obtain the following governing set of coupled differential equations for the interaction-driven dispersion model:
\begin{equation} \label{eq:general_1}
    \frac{dy_i}{dt} = I_i(t) - O_i(t).
\end{equation}
The outflow $O_i(t)$ represents agents leaving state $i$, while the inflow $I_i(t)$ corresponds to agents entering state $i$. Although the model may initially appear to be purely interaction-driven, it can, in fact, accommodate natural recovery mechanisms, as demonstrated, among others, in the example of the SIR and DK models in Section~\ref{subsubsect:sir_daley_DE}.

\subsubsection{Transition mechanism described by graphs}\label{subsubsect:graph_repr}

While the transmission mechanism was rigorously defined in Section~\ref{sec:trans_mech_de}, adopting a more visual perspective helps reveal its underlying structure. Graph representations offer a simplified yet insightful way to uncover this structure. In this section, we introduce two types of graphs that capture different layers of the transition dynamics.

The first type of graph illustrates the possible state transitions. A successor of a state indicates that an agent in the current state may transition to the successor state. Conversely, a predecessor of a state implies that an agent in the current state could have originated from the predecessor state. This is formalized as follows.

\begin{definition}\label{def:trans_graph}
The \textbf{transition graph} of an interaction-driven dispersion model is a directed graph $(V, E)$. The set of vertices $V = \{1, \ldots, N\}$ corresponds to the states of the model. The set of directed edges $E$ contains an edge $(u,v) \in V \times V$ with $u \neq v$ if and only if there exists some state $k \in V$ such that $\beta^v_{uk} > 0$. In this context, a successor $v$ of a vertex $u$ is specifically referred to as a \textbf{transition successor}, implying that an agent in state $u$ can transition to state $v$ through interaction. Conversely, a predecessor $u$ of a vertex $v$ is a \textbf{transition predecessor}, indicating that an agent currently in state $v$ could have originated from state $u$.
\end{definition}

It is important to note that, by its definition, the transition graph excludes self-loops. Cycles, however, may be present. Operationally, we build the transition graph by inspecting the inflow \(I_i\) in \eqref{eq:inflow}: draw a directed edge \(n\to i\) whenever there exists an interaction partner \(k\) with \(\beta_{nk}^{\,i}>0\).

The second type of graph reflects the specific states with which an interaction must occur for a state to transition to a transition successor. It offers crucial insight into the mechanisms that drive the transitions represented in the transition graph.

\begin{definition}\label{def:dep_graph}
The (outflow) \textbf{dependency graph} of an interaction-driven dispersion model is a directed graph $(V, E)$. The set of vertices $V = \{1, \ldots, N\}$ corresponds to the states of the model. The set of directed edges $E$ contains an edge $(u,v) \in V \times V$ if and only if $\beta^u_{uv} > 0$. An edge $(u,v)$ indicates that the outflow from state $u$ is dependent on interactions with agents in state $v$. Successors and predecessors in the dependency graph are specifically referred to as \textbf{dependency successors} and \textbf{dependency predecessors}.
\end{definition}
Analogously, the dependency graph is obtained by inspecting the outflow \(O_i\) in \eqref{eq:outflow}: include a directed edge \(i\to m\) whenever \(\beta^{\,i}_{im}>0\). (Self-loops \(i\to i\) occur when \(\beta^{\,i}_{ii}>0\).)
The combined use of these graph representations facilitates a structural analysis of the interaction-driven dispersion model that we introduced. They will play a pivotal role in Section \ref{sec:results}, where our theoretical findings are stated and proven.

\begin{remark}{\rm
    The interaction-driven dispersion model introduced here admits a compact matrix formulation, which streamlines bookkeeping for large systems and facilitates numerical integration of the model. The full matrix specification is given in Appendix~\ref{app:Bmatrix}.$\hfill\Diamond$}
\end{remark}

\subsection{Example: SIR and Daley-Kendall model}\label{subsect:example_sir_dk}

In this subsection, we present an illustrative example with a threefold aim. First, we show that the well-known SIR and DK models fall within the class of systems described by the interaction-driven dispersion framework. Second, we demonstrate how the natural recovery mechanism in the SIR model can be incorporated into this model. Third, we provide the corresponding graph representations for these specific cases.

\subsubsection{SIR and DK as special cases}\label{subsubsect:sir_daley_DE}
We treat here the (fluid limits of) the SIR model and the DK model. We show how both models can be written in the form of (\ref{eq:general_1}), hence being special cases of our interaction-driven dispersion model. Throughout, we denote by \(y_1, y_2, y_3\) the fractions of susceptible, infectious, and recovered agents in the SIR model; analogously, \(y_1^{*}, y_2^{*}, y_3^{*}\) denote the fractions of ignorants, spreaders, and stiflers in the DK model.

\begin{itemize}
    \item[$\circ$]
The system of coupled differential equations for the SIR model is, with the labeling mentioned above,
\begin{align}
    \frac{dy_1}{dt} &= -\beta y_1 y_2, \label{eq:y1} \\ 
    \frac{dy_2}{dt} &= \beta y_1 y_2 - \gamma y_2, \label{eq:y2} \\
    \frac{dy_3}{dt} &= \gamma y_2 \label{eq:y3},
\end{align}
where $\beta > 0$ is the infection rate, and $\gamma > 0$ is the recovery rate. 

We now show how this set of coupled differential equations aligns with the framework of~(\ref{eq:general_1}), specifically illustrating the integration of the natural recovery mechanism. For state 1, the differential equation, as shown in~(\ref{eq:y1}), consists solely of the outflow component $O_1$ that was defined in~(\ref{eq:outflow}) . This outflow is triggered exclusively by interactions with state-2 agents. Thus, we have $\beta^1_{12} > 0$ (specifically, $\beta^1_{12} = \beta$). The term $y_1 y_2$ appears elsewhere only in the inflow to state 2, as seen in~(\ref{eq:y2}), where its parameter is $\beta^2_{12} = \beta > 0$. The conservation law, as described by~(\ref{eq:cons_law}), is indeed upheld because $\beta^1_{12} = \beta^2_{12} = \beta$. 

The term for natural recovery, namely $\gamma y_2$ in~(\ref{eq:y2}), might initially seem inconsistent with the outflow expression defined in~(\ref{eq:outflow}). However, we can circumvent this by leveraging the conservation law, $y_1 + y_2 + y_3 = 1$, to rewrite it as
\begin{equation}\label{eq:sir_natur_recov}
    \gamma y_2 = \gamma (y_1 + y_2 + y_3) y_2.
\end{equation}
This reformulation allows us to interpret the outflow from state 2 as $y_2 (\beta^2_{21} y_1 + \beta^2_{22} y_2 + \beta^2_{23} y_3)$, where $\beta^2_{21} = \beta^2_{22} = \beta^2_{23} = \gamma$. Under the uniform-mixing assumption, one may read \(\gamma\,y_2(y_1+y_2+y_3)\) as a contact-like outflow in which any
encounter of a state-2 agent carries the same probability of transitioning to state 3. This is a modeling convenience: mechanistically, the term captures spontaneous (natural) recovery rather than an interaction-driven change.\footnote{\footnotesize The natural recovery rate, $\gamma$, may exceed the homogeneous contact rate embedded in the $\beta$ parameters. Nevertheless, by scaling the contact rate by any positive constant and inversely scaling the transition probability, the system's properties remain invariant. This allows any $\gamma > 0$ to be consistently accounted for in the interpretation of the $\beta$ parameters.}
To complete this demonstration, we conclude that $\beta^3_{21} = \beta^3_{22} = \beta^3_{23} = \gamma$, as shown by~(\ref{eq:y3}).

\item[$\circ$]
The system of coupled differential equations for the DK model is given by:
\begin{align}
    \frac{dy_1^*}{dt} &= -\theta y_1^* y_2^*, \\
    \frac{dy_2^*}{dt} &= \theta y_1^* y_2^* - \alpha y_2^* (y_2^* + y_3^*), \label{eq:y*2}\\
    \frac{dy_3^*}{dt} &= \alpha y_2^* (y_2^* + y_3^*),
\end{align}
with $y_1^*$, $y_2^*$, and $y_3^*$ having the meaning defined above, $\theta$ encapsulating the rate at which Ignorants become Spreaders upon contact with Spreaders, and $\alpha$ the rate at which Spreaders become Stiflers due to interactions with other Spreaders or Stiflers.

\end{itemize}

The key distinction between the SIR and DK models is evident in the outflow dynamics of $y_2$. The SIR model allows a state-2 agent to transition to state 3 upon contact with a state-1 agent, a mechanism explicitly excluded in the DK model; see (\ref{eq:y2}) and (\ref{eq:y*2}).

\subsubsection{Graph representation and matrix formulation of SIR and DK}
We conclude this section by demonstrating how the SIR and DK models can be represented using the graph representation of our interaction-driven dispersion model, as introduced in Section \ref{subsubsect:graph_repr}.

For the SIR model we first identify its transition graph, using Definition \ref{def:trans_graph}, focusing solely on the inflow $I_i$. State 2 is a transition successor of state 1, given that $\beta^2_{12} > 0$ (see Section \ref{subsubsect:sir_daley_DE}). This implies an edge $(1,2)$ in the transition graph. Furthermore, since $\beta^3_{2j} > 0$ for $j =1,2,3$, state 3 is a transition successor of state 2. This corresponds to an edge $(2,3)$. The transition graph of the SIR model is thus depicted by Figure \ref{fig:transition_graph_sir_dk}. Through a similar line of reasoning, we conclude that the transition graph for the DK model is identical to that of the SIR model.

For the SIR model's dependency graph, as introduced in Definition \ref{def:dep_graph}, we identify the following dependencies by analyzing the outflow $O_i$. An edge $(1,2)$ exists because $\beta^1_{12} > 0$ (see Section \ref{subsubsect:sir_daley_DE}), indicating that outflow from State 1 depends on interactions with State 2 agents. Furthermore, since $\beta^2_{21} = \beta^2_{22} = \beta^2_{23} = \gamma > 0$, state 2 serves as a dependency predecessor for all states in the system (i.e., edges $(2,1), (2,2), (2,3)$ exist), highlighting that its outflow depends on interactions with agents from any state, including itself. This illustrates the impact of the natural recovery mechanism within our interaction-driven dispersion model. The dependency graph of the SIR model is depicted in Figure \ref{fig:dependency_graph_sir}. By an identical line of reasoning, the DK model's dependency graph is shown in Figure \ref{fig:dependency_graph_dk}. 

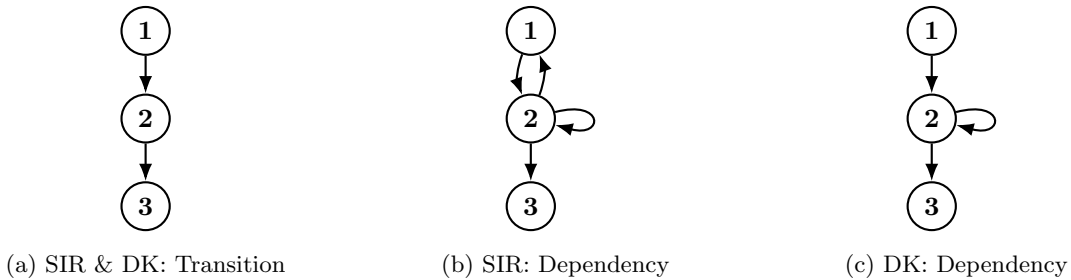
\begin{figure}[H]
    \centering
    \begin{subfigure}[b]{0.3\textwidth} 
        \centering
        \begin{tikzpicture}[
            mynode/.style={circle, draw, thick, minimum size=0.2cm, font=\bfseries}, 
            myarrow/.style={-Latex, thick, draw=black},
            bidir_arrow_1to2/.style={myarrow, bend right=20},
            bidir_arrow_2to1/.style={myarrow, bend right=20},
            self_loop_arrow/.style={myarrow, loop right, looseness=12}            
        ]
            \node[mynode] (N1) {1};
            \node[mynode, below=0.5cm of N1] (N2) {2};
            \node[mynode, below=0.5cm of N2] (N3) {3};

            \draw[myarrow] (N1) -- (N2); 
            \draw[myarrow] (N2) -- (N3); 
        \end{tikzpicture}
        \caption{SIR \& DK: Transition} 
        \label{fig:transition_graph_sir_dk}
    \end{subfigure}
    \hfill
    \begin{subfigure}[b]{0.3\textwidth}
        \centering
        \begin{tikzpicture}[
            mynode/.style={circle, draw, thick, minimum size=0.2cm, font=\bfseries}, 
            myarrow/.style={-Latex, thick, draw=black},
            bidir_arrow_1to2/.style={myarrow, bend right=20},
            bidir_arrow_2to1/.style={myarrow, bend right=20},
            self_loop_arrow/.style={myarrow, loop right, looseness=10}
        ]

            \node[mynode] (N1) {1};
            \node[mynode, below=0.5cm of N1] (N2) {2};
            \node[mynode, below=0.5cm of N2] (N3) {3};

            \draw[bidir_arrow_1to2] (N1) to (N2);
            \draw[bidir_arrow_2to1] (N2) to (N1);
            \draw[self_loop_arrow] (N2) to (N2);
            \draw[myarrow] (N2) -- (N3);

        \end{tikzpicture}
        \caption{SIR: Dependency} 
        \label{fig:dependency_graph_sir}
    \end{subfigure}
    \hfill
    \begin{subfigure}[b]{0.3\textwidth}
        \centering
        \begin{tikzpicture}[
            mynode/.style={circle, draw, thick, minimum size=0.2cm, font=\bfseries}, 
            myarrow/.style={-Latex, thick, draw=black},
            bidir_arrow_1to2/.style={myarrow, bend right=20},
            bidir_arrow_2to1/.style={myarrow, bend right=20},
            self_loop_arrow/.style={myarrow, loop right, looseness=10}
        ]

            \node[mynode] (N1) {1};
            \node[mynode, below=0.5cm of N1] (N2) {2};
            \node[mynode, below=0.5cm of N2] (N3) {3};

            \draw[myarrow] (N1) -- (N2);
            \draw[self_loop_arrow] (N2) to (N2);
            \draw[myarrow] (N2) -- (N3);

        \end{tikzpicture}
        \caption{DK: Dependency} 
        \label{fig:dependency_graph_dk}
    \end{subfigure}
    \caption{Transition and Dependency graph representation of the SIR and DK models} 
    \label{fig:graphs_sir}
\end{figure}
Figure \ref{fig:graphs_sir} reveals a critical distinction between the SIR and DK models: the presence or absence of the edge $(2,1)$ in their dependency graphs. This is surprising because, while the SIR model incorporates a natural recovery mechanism -- a process typically considered independent of specific interactions -- this graphical representation depicts it as an interaction-dependent outflow involving state-1 agents. Such a precise mapping within our unified framework highlights that these two models, despite their differing recovery conceptualizations, are fundamentally more similar than commonly perceived. We have provided the matrix formulation of the SIR and DK models in Appendix~\ref{app:Bmatrix}.

\section{Theoretical results}\label{sec:results}
This section, presenting our theoretical results, is organized into three parts. First, we establish that the system (\ref{eq:general_1}) is {\it asymptotically stable}  — that is, all state variables converge to a limit over time and at the same time its derivatives converge to zero   — under an explicit assumption on the structure of the underlying interaction-driven dispersion model. 
Second, we establish our main theorem, which --- under suitable conditions --- characterizes long-term behavior by providing a precise relationship between vanishing states (whose state variable converges to zero) and persistent states (whose state variable converges to a strictly positive limit). This result yields a complete classification of the asymptotic outcomes for all states in the system. 
In the proof, a crucial role is played by the $L^1$-integrability of the state variables — more specifically, by the condition whether or not $\int_{0}^{\infty} y_i(t)\,dt$ is finite.
Third, building on these findings, we discuss informally that all vanishing states decay to zero at an exponential rate. Moreover, we (formally) demonstrate that if the aforementioned suitable conditions are not satisfied, then this exponential decay may no longer hold.

\subsection{Asymptotic stability}
The main result of this subsection is Theorem \ref{MR}, which states the asymptotic stability of the interaction-driven dispersion model. We prove this under the condition that the underlying transition graph, as defined in Section \ref{subsubsect:graph_repr}, must be a DAG (see Definition \ref{def:dag}).

\begin{assumption}[DAG assumption]\label{ass:trans_dag}
    The underlying transition graph of the interaction-driven dispersion model is a DAG.
\end{assumption}

The interpretation of Assumption \ref{ass:trans_dag} is that the transition graph must preclude the existence of cycles. Intuitively, cycles could enable indefinite circulation of flows within the system, potentially hindering the convergence of the state variables. In fact, a DAG structure implies the existence of at least one source and one sink (see Definition \ref{def:sink_source}), ensuring a unidirectional flow of agents. This property is formalized by the evident Lemma \ref{lemma:strongly_connected}, with its proof being provided in Appendix~\ref{app:proofs} for completeness.

\begin{lemma}[DAG has a sink and a source]\label{lemma:strongly_connected}
    A directed graph $(V,E)$ that is a DAG has at least one sink and at least one source.
\end{lemma}

A prerequisite for proving the main theorem of this section is to demonstrate that the state fractions, $y_i$, remain confined within $[0,1]$. The boundedness, though seemingly obvious, requires a non-trivial formal proof. Lemma \ref{lemma:bounded_y} establishes the property for our class of interaction-driven dispersion models. The proof is provided in Appendix~\ref{app:proofs}.

\begin{lemma}[Confinement of state variables]\label{lemma:bounded_y}
    In the interaction-driven dispersion model we have $0 < y_i(t) < 1$, for $i = 1, \ldots, N$ and for any $t \geq 0$. 
\end{lemma}

A crucial consequence of the conservation law, which informally entails that no agents are added to the population from external sources, is that the inflow to any given state must always be less than or equal to the total outflow of its predecessors. This property, while intuitively straightforward, proves to be a critical inequality for deriving our theoretical results. The proof is provided in Appendix~\ref{app:proofs}. We denote by ${\mathscr P}_i$ the set of transition predecessors of state $i$.

\begin{lemma}[Inflow-outflow inequality]\label{lem:I_leq_O} 
 The inflow to state $i$ is less than or equal to the sum of the outflows from its transition predecessors:
\[ I_i(t) \leq \sum_{k \in {\mathscr P}_i} O_k (t). \]
\end{lemma}

The following main theorem demonstrates that the system ultimately ceases its dynamic change, reaching an asymptotically stable state, provided its underlying transition graph is acyclic.

\begin{theorem}[Asymptotic stability]\label{MR}
   Under Assumption \ref{ass:trans_dag}, the system reaches a steady state:  for each state $i \in \{1, \ldots, N\}$, the state variable $y_i(t)$ converges to a constant value $y_i^* \in [0,1]$ as $t \to \infty$. Moreover, 
    \[
        \lim_{t \to \infty} y_i(t) = y_i^*, \quad \lim_{t \to \infty} \frac{dy_i}{dt} = 0, \quad \text{and} \quad \lim_{t \to \infty} O_i(t) = \lim_{t \to \infty} I_i(t) = 0
    \]
\end{theorem}
\begin{proof}

    We prove the theorem by induction on the vertices of the state transition graph, $G=(V,E)$, which is a DAG by Assumption \ref{ass:trans_dag}. The vertices are labeled $1, 2, \dots, N$ according to a topological sort. This ordering, which is possible by Assumption \ref{ass:trans_dag}, guarantees that if a transition from state $j$ to state $k$ is possible (see Definition~\ref{def:trans_graph}), then their labels must satisfy $j < k$. The core of the proof relies on the application of the monotone convergence theorem.

    For the proof, we first define a set of non-increasing functions, which uses that the DAG assumption implies that there is a `directional flow' from the source to the sink vertex. 
To this end, we define
    \[
    Z_m(t): = \sum_{k \in V_m} y_{k}(t),
    \]
    with $V_m = \{1, 2, \ldots, m\}$, the set of the first $m$ states.
    Our first objective is to show that each $Z_m(t)$ is non-increasing in $t$. We do this by assessing its time derivative:
    $$
    \frac{dZ_m(t)}{dt} = \sum_{k \in V_m} \frac{dy_{k}}{dt} = \sum_{k \in V_m} \left( I_{k}(t) - O_{k}(t) \right).
    $$
    Let us analyze the inflow and outflow to and from $V_m$. 
    \begin{itemize}
    \item[$\circ$]
      \textit{In-flow to $V_m$:} An in-flow to a node $k \in V_m$ is represented by a term $\beta^k_{pj} y_p y_j$ in $I_{k}$, which implies an edge $p \to k$, $p \ne k$. Since the states are topologically ordered and $V_m$ consists of the first $m$ states in that order, any transition predecessor $p$ of $k$ must also be in $V_m$. Therefore, there is no flow into $V_m$ from outside (i.e., from the complement of $V_m$, denoted by $V_m^c$).
    \item[$\circ$]
      \textit{Out-flow from $V_m$:} The out-flow from $V_m$ to $V_m^c$ is the sum of all flows from nodes $k \in V_m$ to nodes $p \notin V_m$.
    \end{itemize}     
    The derivative of $Z_m(t)$ simplifies to the negative of the total out-flow from $V_m$ to $V_m^c$:
    \[
    \frac{dZ_m(t)}{dt} = - \sum_{k=1}^{m} \sum_{p \notin V_m} \sum_{j=1}^{N} \beta^p_{kj} y_{k}(t) y_j(t)
    \]
    Since all $\beta$ coefficients and all $y_i$ values are non-negative, the sum on the right-hand side is non-negative, which yields
    $$
    \frac{dZ_m(t)}{dt} \le 0
    $$
    This proves that for each $m=1, \ldots, N$, the function $Z_m(t)$ is non-increasing. Now that we have proven this property, we proceed by the actual induction proof of our claim. 
    
    \begin{itemize}\item[$\circ$]\textit{Base case.} 
We begin by studying state $m=1$, which is a source vertex in the transition graph. From the initial part of the proof we know that the function $Z_1(t)$ (which equals $y_1(t)$) is non-increasing, and from Lemma \ref{lemma:bounded_y} that it is bounded between $0$ and $1$. Therefore, by the monotone convergence theorem, the limit of $Z_1(t)$, as $t \to \infty$,  must exist. 
The outflow, $O_1(t)$, is by definition a non-negative function. Consequently, $\frac{dy_1}{dt} = -O_1(t)$ is a non-positive function; see \eqref{eq:outflow}. As $O_1(t)$ is composed of bounded state variables $y_k(t)$ (as they lie between $0$ and $1$) and finite parameters, $O_1(t)$ is a bounded function. Therefore, by the monotone convergence theorem, the limit of $\frac{dy_1}{dt}$ must exist; see \eqref{eq:general_1}. Combining this with the fact that we have already established that the limit of $y_1(t)$ exists, its time derivative $\frac{dy_1}{dt}$ must vanish as $t \to \infty$.

We thus conclude that
    \[  \lim_{t \to \infty} Z_1(t) = Z_1^*, \quad
        \lim_{t \to \infty} y_1(t) = y_1^*, \quad \lim_{t \to \infty} \frac{dy_1}{dt} = 0, \quad \text{and} \quad \lim_{t \to \infty} O_1(t) = \lim_{t \to \infty} I_1(t) = 0,
    \]
with $Z_1^* \in [0,1]$ and $y_1^* \in [0,1]$.    

    \item[$\circ$]\textit{Induction step.} 
For the induction step, we assume, for any state $n \in \{1, \ldots, m-1\}$ (where $m-1 < N$), that
\begin{equation}\label{eq:induce_ass}
    \lim_{t \to \infty} Z_n(t) = Z_n^*, \quad \lim_{t \to \infty} y_n(t) = y_n^*, \quad \lim_{t \to \infty} \frac{dy_n}{dt} = 0, \quad \text{and} \quad \lim_{t \to \infty} O_n(t) = \lim_{t \to \infty} I_n(t) = 0,
\end{equation}
where $Z_n^* \in [0,n]$ and $y_n^* \in [0,1]$. Our goal is to prove that these same conditions also hold for state $m$.

From the first part of the proof, we know that $Z_m(t)$ is a non-increasing function. Additionally, by Lemma \ref{lemma:bounded_y}, $Z_m(t)$ is bounded between $0$ and $m$. Thus, by the monotone convergence theorem, we conclude that a limit $Z_m^* \in [0,m]$ exists such that $\lim_{t\to \infty} Z_m(t) = Z^*_m$.

Observe $y_m(t) = Z_m(t) - Z_{m-1}(t)$. Because $Z_{m-1}(t)$ and $Z_m(t)$ are monotone and bounded, their limits (say $Z_{m-1}^*$ and $Z_m^*$, respectively) exist. Hence, their difference has a limit, too. Thus, $\lim_{t \to \infty} y_m(t) = y_m^*$,
which lies in $[0,1]$  by virtue of Lemma \ref{lemma:bounded_y}.

We are going to prove that $\lim_{t \to \infty} I_m(t) = 0$, which will then be used to prove that $O_m(t) = 0$. First, recall Lemma \ref{lem:I_leq_O}:
\[I_m(t) \leq \sum_{k \in {\mathscr P}_m} O_k (t) = \sum_{k=1}^{m-1} O_k(t),\]
where the summation limit $k < m$ reflects that any inflow into state $m$ must originate from a state `earlier' in the topological sort of the transition graph.

By the induction hypothesis (\ref{eq:induce_ass}), we hav for any $k \in \{1, \ldots, m-1\}$, that $\lim_{t \to \infty} O_k(t) = 0$. Since $I_m(t)$ is bounded above by a finite sum of terms that each tend to zero, we conclude that $\lim_{t \to \infty} I_m(t) = 0$. 

We thus have established the following: $y_m(t)$ is bounded between $0$ and $1$, its inflow $I_m(t)$ has limit $0$, and its outflow $O_m(t)$ is a non-negative function. By Lemma \ref{lemma:limit_combi}, a technical result stated and proven in Appendix~\ref{app:techn}, these conditions imply that the limit of the outflow, $\lim_{t\to\infty} O_m(t)$, must be equal to $0$ and that $\lim_{t\to\infty} \frac{dy_m}{dt} = 0$.

    \end{itemize} We have thus proven the desired.
\end{proof}

\subsection{Interdependence of vanishing and persistent asymptotic states}

This section introduces our main theoretical contribution: a graph-theoretic property that characterizes the relationship between vanishing states (i.e., those in which the state variables converge to zero) and persistent states (i.e., those in which the state variables converge to a strictly positive value). This property enables a complete classification of the asymptotic outcome of state variables across a broad class of interaction-driven dispersion models, including well-known examples such as the SIR and DK models, as well as many of their extensions.

\begin{definition} Let \(y_i(t)\) denote the state variable for state \(i\). \begin{itemize}
   \item[$\circ$] We call state \(i\) \textbf{vanishing} if \(\lim_{t\to\infty} y_i(t)=0\) and \textbf{persistent} if \(\lim_{t\to\infty} y_i(t)>0\).
   \item[$\circ$] The \textbf{asymptotic characterization} of a state specifies whether it is vanishing or persistent. The asymptotic characterisation of a model instance is the assignment of these labels to all states in the interaction-driven dispersion model.
\end{itemize}

\end{definition}

In this subsection, we again impose Assumption \ref{ass:trans_dag}, stating that the transition graph be a DAG. With this in mind, we first establish results concerning the integrability of the vanishing fractions $y_i$, which provide us with essential tools and insights. The proof of this integrability Theorem \ref{th:integrable} beautifully illustrates the conditions on the dependency graph for it to hold. This theorem also offers an interpretation regarding decay rates (Conjecture \ref{conj:exp_decay}), as a non-integrable fraction inherently implies a slower decay rate than an integrable one.
Building upon the integrability theorem, we then present the key finding of this paper: 
\begin{quotation}
 \textit{a state is persistent if and only if all of its dependency successors are vanishing.}   
\end{quotation}
This theorem is powerful and practical. It not only provides insight into the rules governing asymptotic states but also, for models meeting its conditions, allows for a complete asymptotic characterization of all states. To facilitate its application, we will provide a constructive method for attaining this characterization.

\subsubsection{Dependency graph assumption}

Before proceeding with our results, we introduce and discuss the crucial  assumption regarding the dependency graph.

\begin{assumption}[DAG$^{-}$ assumption]\label{ass:depend_dag+}
Every cycle in the underlying dependency graph contains at least one state with a persistent dependency successor.
\end{assumption}

An examination of Assumption \ref{ass:depend_dag+} reveals that it is satisfied if the dependency graph contains no cycles, entailing it is met by any DAG. However, it clearly also covers instances that are no DAGs. Therefore, we refer to Assumption \ref{ass:depend_dag+} as the DAG$^{-}$ assumption, the minus sign reflecting that we impose less demanding requirements.

\begin{remark}{\rm 
In this remark we provide an interpretation of Assumption \ref{ass:depend_dag+}, focusing on the requirement that every cycle must include at least one state with a persistent dependency successor. 

To build intuition, we informally assess the effect of cycles in dependency graphs, beginning with the special case of self-loops, but the reasoning extends to general cycles as well.
Recall that an outgoing edge in the dependency graph indicates that a state's outflow is proportional to the fraction of the state it targets. A self-loop, then, implies that a state’s outflow depends on its own fraction, causing the outflow (and hence the decay of the state) to slow down over time. This is in contrast to cases where a state has a persistent dependency successor, which ensures that a state's fraction decreases proportionally to a `minimum rate'.
Intuitively, a self-loop leads to sub-exponential decay, introducing a bottleneck that slows the overall convergence of the system. This slower decay can propagate through the graph, impeding the system’s ability to rapidly approach a stable configuration (see Example \ref{example:3-cycle}).

However, the presence of such cycles does not impede overall rapid convergence as long as a faster decay mechanism is present. This faster decay is assured by the existence of a persistent dependency successor, which is precisely the condition imposed by the DAG$^{-}$ assumption.}$\hfill\Diamond$
\end{remark}

Thus, for interaction-driven dispersion models with an underlying DAG transition graph, the DAG$^{-}$ assumption draws conservative borderlines that distinguish model behavior into two distinct realms. On one side are systems that converge exponentially and exhibit the asymptotic rules derived in this paper. On the other side lie systems that converge more slowly (hypothesized to be inversely proportional to time, as suggested by Example \ref{example:3-cycle}), where these rules no longer apply.

\begin{remark} {\rm 
   The DAG$^{-}$ assumption may initially appear circular: it is invoked to derive the asymptotic characterization of system states, yet it presupposes prior knowledge of which states are persistent. However, this assumption accommodates a significantly broader class of dependency graphs that include cycles. The key lies in Assumption~\ref{ass:trans_dag} — that the transition graph is a DAG — which, by Lemma~\ref{lemma:strongly_connected}, ensures the existence of sink states. Since sink states are necessarily persistent, they provide a concrete starting point for the analysis and resolve the apparent circularity. $\hfill\Diamond$}
\end{remark}

\subsubsection{Example: SIR and DK models as instances of DAG$^{-}$} \label{subsubsect:sir_dk_dag-}

To illustrate Assumption \ref{ass:depend_dag+}, we demonstrate how to verify whether the dependency graphs of the SIR and DK models satisfy the DAG$^{-}$ assumption. This analysis continues the example introduced in Subsection \ref{subsect:example_sir_dk}.

We begin by establishing that the transition graphs of the SIR and DK models, as depicted in Figure \ref{fig:transition_graph_sir_dk}, are clearly DAGs, thereby satisfying Assumption \ref{ass:trans_dag}.

Now, examining the dependency graph of the SIR model, we find that it contains cycles, which means the DAG$^{-}$ assumption must be carefully verified. As shown in Figure \ref{fig:dependency_graph_sir_dag+}, we detect two distinct cycles: $C_1: 2 \to 2$ (highlighted in red) and $C_2: 2 \to 1 \to 2$ (highlighted in blue). We know that the transition sink, state 3, is a persistent state. Furthermore, Figure \ref{fig:dependency_graph_sir_dag+} confirms that state 3 is a dependency successor of state 2. Therefore, both $C_1$ and $C_2$ contain a state (state 2) that has a persistent dependency successor (state 3), leading us to conclude that the SIR model satisfies the DAG$^{-}$ assumption. A similar line of reasoning confirms that the DK model also satisfies the DAG$^{-}$ assumption.

The SIR model's natural recovery mechanism is inherently compatible with the DAG$^{-}$ assumption. This is because a state that allows for natural recovery in the dependency graph has every other state as a successor. This implies two key points: 1) its graph inherently contains cycles involving the state that allows for natural recovery, and 2) the natural recovery state always has the system's persistent sink state as a dependency successor. Consequently, cycles involving states with natural recovery mechanisms naturally fulfill the DAG$^{-}$ assumption.

\begin{figure}[H]
    \centering
    \begin{subfigure}[b]{0.48\textwidth}
        \centering
        \begin{tikzpicture}[
            mynode/.style={circle, draw, thick, minimum size=0.2cm, font=\bfseries}, 
            myarrow/.style={-Latex, thick, draw=black},
            bidir_arrow_1to2/.style={myarrow, bend right=20,blue},
            bidir_arrow_2to1/.style={myarrow, bend right=20,blue},
            self_loop_arrow/.style={myarrow, loop right, looseness=10,red}
        ]

            \node[mynode] (N1) {1};
            \node[mynode, below=0.5cm of N1] (N2) {2};
            \node[mynode, below=0.5cm of N2] (N3) {3};

            \draw[bidir_arrow_1to2] (N1) to (N2);
            \draw[bidir_arrow_2to1] (N2) to (N1);
            \draw[self_loop_arrow] (N2) to (N2);
            \draw[myarrow] (N2) -- (N3);

        \end{tikzpicture}
        \caption{SIR: Dependency graph with two cycles} 
        \label{fig:dependency_graph_sir_dag+}
    \end{subfigure}
    \hfill
    \begin{subfigure}[b]{0.48\textwidth}
        \centering
        \begin{tikzpicture}[
            mynode/.style={circle, draw, thick, minimum size=0.2cm, font=\bfseries}, 
            myarrow/.style={-Latex, thick, draw=black},
            bidir_arrow_1to2/.style={myarrow, bend right=20},
            bidir_arrow_2to1/.style={myarrow, bend right=20},
            self_loop_arrow/.style={myarrow, loop right, looseness=10,green}
        ]

            \node[mynode] (N1) {1};
            \node[mynode, below=0.5cm of N1] (N2) {2};
            \node[mynode, below=0.5cm of N2] (N3) {3};

            \draw[myarrow] (N1) -- (N2);
            \draw[self_loop_arrow] (N2) to (N2);
            \draw[myarrow] (N2) -- (N3);

        \end{tikzpicture}
        \caption{DK: Dependency graph with one cycle} 
        \label{fig:dependency_graph_dk_dag+}
    \end{subfigure}
    \caption{Dependency graphs of the SIR and DK models with highlighted cycles in blue, red and green} 
    \label{fig:graphs_sir_dag+}
\end{figure}
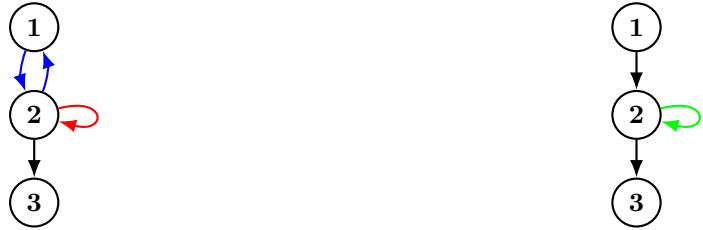

\subsubsection{Integrability results and interdependence of asymptotic states}

We proceed by showing that for interaction-driven dispersion models with DAG transition graphs, and under the additional DAG$^{-}$ assumption on the dependency graph, the fractions $y_i$ are indeed integrable. To establish this, we first present a lemma demonstrating the integrability of each state's inflow and outflow.

The integrability of inflow and outflow is intuitively expected due to the DAG assumption on the transition graph. This assumption ensures that the flow of agents, which is a finite and bounded amount, moves in only one direction. This unidirectional flow implies that both the inflow to any state and its subsequent outflow must be finite. We formalize this intuition below and provide the proof in Appendix~\ref{app:proofs}.

\begin{lemma}[Integrability of inflow and outflow]\label{lem:I_and_O_finite}
    Under Assumption \ref{ass:trans_dag}, for any state $i \in \{1, \ldots, N\}$, the integrals $\int_0^\infty I_i(t) \,dt$ and $\int_0^\infty O_i(t) \,dt$ are finite.
\end{lemma}

Now that we have established some model's properties, we are ready to present the integrability theorem. Its proof is particularly insightful, demonstrating how the DAG$^{-}$ assumption fits neatly into the conditions for integrability.

\begin{theorem}[Integrability of vanishing states]\label{th:integrable}
Under Assumptions~\ref{ass:trans_dag} and~\ref{ass:depend_dag+}, for any state $i \in \{1, \dots, N\}$, if $y_i$ vanishes over time $($i.e., \(\lim_{t \to \infty} y_i(t) = 0\)$)$, then its time integral is finite $($i.e., \(\int_0^\infty y_i(t)\,dt < \infty\)$)$.
\end{theorem}
\begin{proof}
We will prove this by contradiction. To this end, we assume the set of vanishing but non-integrable states $A$ is not empty; here, $A$ is defined as
\[ A := \left\{k \: {\Large\mid}\: \lim_{t\to\infty} y_k(t) = 0 \wedge \int_0^\infty y_k(t)\,dt = \infty \right\}. \]
Our first goal is to prove the following claim:
\begin{quotation}
    {\it Any successor $k$ of state $j \in A$ must be in $A$ itself.}
\end{quotation}
If this claim holds, then the existence of any $j_1 \in A$ necessitates an infinite dependency path $j_1 \to j_2 \to \ldots$, where all $j_i \in A$. However, by the DAG$^{-}$ assumption on the dependency graph, any cycle in the dependency graph contains a state $j$ with a persistent successor $k$. This means such cycles inherently do not satisfy the aforementioned property of set $A$, forbidding $A$ from being non-empty. This directly implies that any vanishing state $k$ must be integrable.

Let us now proceed to prove that if $A$ is nonempty, then every state $i \in A$ must have a successor $k \in A$ in the dependency graph. To that end, let $i$ be an arbitrary state in $A$.

    \begin{itemize}
        \item[$\circ$]
\textit{Step 1: There must be a successor $k$ of $i$ that is non-integrable.}
We focus on the logarithmic derivative of $y_i(t)$:
$$\frac{d}{dt} \ln y_i(t) = \frac{I_i(t)}{y_i(t)} - \frac{O_i(t)}{y_i(t)} = \frac{I_i(t)}{y_i(t)} - \frac{y_i(t) \sum_{j}^N \beta^i_{ij}y_j(t)}{y_i(t)} = \frac{I_i(t)}{y_i(t)} - S_i(t),$$
where 
\begin{equation}
    \label{defS}
    S_i(t) := \sum_{j=1}^N \beta^i_{ij}y_j(t)    
\end{equation}
represents the total outflow rate per unit of $y_i(t)$.
Integrating both sides from $0$ to $T$ yields
\begin{equation}\label{eq:log}
    \ln y_i(T) - \ln y_i(0) = \int_0^T \frac{I_i(t)}{y_i(t)}\,dt - \int_0^T S_i(t) \,dt.
\end{equation}
As $T \to \infty$, since $i \in A$, we know $\lim_{T\to\infty} y_i(T) = 0$, which implies $\lim_{T\to\infty} \ln y_i(T) \to -\infty$. Therefore, the left-hand side of~(\ref{eq:log}) tends to $-\infty$. Given that ${I_i(t)}/{y_i(t)} \geq 0$ for all $t$, for the left-hand side to diverge to $-\infty$, we necessarily have that
$$\int_0^\infty S_i(t) \,dt = \infty.$$
Recalling the definition \eqref{defS}, we find that $\beta^i_{ij} \ge 0$ implies that, for $\int_0^\infty S_i(t) \,dt$ to diverge, there must exist at least one successor state $k$ of $i$ in the dependency graph (i.e., $\beta^i_{ik} > 0$) such that its integral $\int_0^\infty y_k(t)\, dt$ is infinite. Let $k$ be the state with a non-integrable state variable $y_k$.

    \item[$\circ$]\textit{Step 2: The successor $k$ must also be in $A$.}
Assume, in order to produce a contradiction, that state $k$ is persistent. This means $\lim_{t \to \infty} y_k(t) = y_k^*$ for some $y_k^* > 0$, which implies $k \notin A$. Because $y_k(t)$ approaches a positive constant $y_k^*$, there must exist a constant $C > 0$ and a time $T_0 > 0$ such that for all $t > T_0$, $y_k(t) > C$.

Bearing in mind \eqref{defS}, since we know $\beta^i_{ik} > 0$ (because $k$ is a successor of $i$), we conclude that $S_i(t) \ge \beta^i_{ik} y_k(t)$. Therefore, for $t > T_0$, we have
$$S_i(t) > \beta^i_{ik} C.$$
Now, considering the outgoing integral for state $i$:
$$O_i(t) = y_i(t) S_i(t) > (\beta^i_{ik} C)\, y_i(t) \quad \text{for } t > T_0.$$
Integrating both sides over the interval $[0, \infty)$ (or from $T_0$ to $\infty$ and adjusting for the initial segment), we thus obtain
$$(\beta^i_{ik} C) \int_0^\infty y_i(t) \,dt < \int_0^\infty O_i(t) \,dt.$$
As we have assumed that the transition graph is a DAG, we can apply Lemma \ref{lem:I_and_O_finite}, which states that the integral $\int_0^\infty O_i(t) \,dt$ is finite. Consequently, the left-hand side of the inequality must also be finite, which implies $\int_0^\infty y_i(t)\, dt < \infty$. This directly contradicts our initial assumption that $i \in A$ (since $i \in A$ implies $\int_0^\infty y_i(t) \,dt = \infty$).
Therefore, our assumption that $k$ is persistent must be false. Since we already established $\int_0^\infty y_k(t)\, dt = \infty$ in Step 1, we conclude that $k \in A$.

Since the reasoning above applies to any state in $A$, we have shown that if $A$ is non-empty, one can indeed construct an infinite dependency path $j_1 \to j_2 \to \ldots$ where all $j_i \in A$, as established at the beginning of the proof. This leads to the logical contradiction previously detailed, as such a path is impossible in a finite graph that is DAG$^{-}$. Thus, $A$ must be empty, proving that any vanishing state must be integrable.
\end{itemize}
This completes the proof.
\end{proof}

With the integrability theorem now established --- a result that informally ascertains a rapid speed of convergence for the system --- we are ready to present the main theorem of this paper. This theorem rigorously defines a strict interdependent structure governing the asymptotic characterization of the system's state variables (i.e., whether they vanish or persist).

\begin{theorem} \label{thm:vanish_persist_basic}
Under Assumptions \ref{ass:trans_dag} and \ref{ass:depend_dag+}, a state $i$ is persistent $($i.e., $\lim_{t\to\infty} y_i(t) > 0$$)$ if and only if all of its successors $j$ in the dependency graph are vanishing $($i.e., for all $j$ such that $\beta^i_{ij} > 0$, we have $\lim_{t\to\infty} y_j(t) = 0$$)$.
\end{theorem}
\begin{proof}
We separately prove the two implications.
\begin{itemize}
    \item[$\circ$]
\textit{Proof of the forward implication.}
We  prove that if state $i$ is persistent, all of its dependency successors must be vanishing. We do this by contradiction. Assume state $i$ is persistent, but that it has at least one dependency successor, $j$, that is also persistent. This gives us two conditions:
\begin{enumerate}
    \item $\lim_{t\to\infty} y_i(t) = y_i^* > 0$.
    \item There exists a state $j$ where $\beta^i_{ij} > 0$ and $\lim_{t\to\infty} y_j(t) = y_j^* > 0$.
\end{enumerate}
Now, consider the outflow of state $i$, $O_i(t) = y_i(t) \sum_{k=1}^N \beta^i_{ik} y_k(t)$. We evaluate 
\[ \lim_{t\to\infty} O_i(t) = \left(\lim_{t\to\infty} y_i(t)\right) \cdot \left(\lim_{t\to\infty} \sum_{k=1}^N \beta^i_{ik} y_k(t)\right) = y_i^* \cdot \left(\sum_{k=1}^N \beta^{i}_{ik} y_k^*\right). \]
Given that  $\beta^i_{ij} > 0$ and $y_j^* > 0$ (from condition 2), it directly follows that $\sum_{k=1}^N \beta^{i}_{ik} y_k^* > 0$. Also we have $y_i^* > 0$ (from condition 1), which leads us to the conclusion $\lim_{t \to \infty} O_i(t) > 0$.

However, this conclusion directly contradicts Theorem \ref{MR}, which states that if the transition graph is DAG, then $\lim_{t \to \infty} O_i(t) = 0$. Hence, our initial assumption that state $i$ has a persistent dependency successor must be false. This means that the state variables of all dependency successors of $i$ must indeed be vanishing.
  \item[$\circ$]
\textit{Proof of the backward implication.}
We continue by proving that if all of state $i$'s dependency successors are vanishing, then state $i$ must be persistent, again by contradiction. Assume that all of $i$'s dependency successors vanish, and also assume that $i$ itself vanishes. This gives us two conditions:
\begin{enumerate}
    \item $\lim_{t\to\infty} y_i(t) = 0$.
    \item For all $j$ such that $\beta^i_{ij} > 0$, we have $\lim_{t\to\infty} y_j(t) = 0$.
\end{enumerate}
We can analyze the logarithmic derivative of $y_i(t)$ again, arriving at the same expression as found in~(\ref{eq:log}). Since $\lim_{t \to \infty} y_i(t) = 0$ (from condition 1), the left-hand side of~(\ref{eq:log}) diverges to $-\infty$. 
Regarding the right-hand side of (\ref{eq:log}) this entails that $\int_0^\infty S_i(t) \, dt=\infty$. Hence,
\begin{equation}\label{eq:infty}
    \int_0^\infty S_i(t)\,dt = \int_0^\infty \sum_{k=1}^N \beta^i_{ik} y_k(t)\,dt = \sum_{k=1}^N \beta^i_{ik} \int_0^\infty y_k(t)\,dt = \infty.
\end{equation}
 
By our assumption (condition 2), for all $k$ such that $\beta^i_{ik} > 0$, we know that $\lim_{t \to \infty} y_k(t) = 0$. Furthermore, by the integrability theorem (Theorem \ref{th:integrable}), we have that $\int_0^\infty y_k(t)\,dt < \infty$ for all such $k$. Therefore, $\sum_{k=1}^N \beta^i_{ik} \int_0^\infty y_k(t)\,dt$ must be finite, as it is a finite sum of finite non-negative terms. This directly contradicts (\ref{eq:infty}).
We must, therefore, reject our assumption that $i$ vanishes. We conclude that $i$ must be a persistent state.

\end{itemize}
This completes the proof. 
\end{proof}

Theorem \ref{thm:vanish_persist_basic} provides a strict rule for the interdependence of asymptotic states. As an illustration, consider the subclass of interaction-driven dispersion models where the transition graph and dependency graph are identical, are DAGs and possess symmetry. For such models, Theorem \ref{thm:vanish_persist_basic} reveals that vanishing and persistent states must exhibit an alternating pattern, as exemplified in Figure~\ref{fig:interdependence_symmetric}.

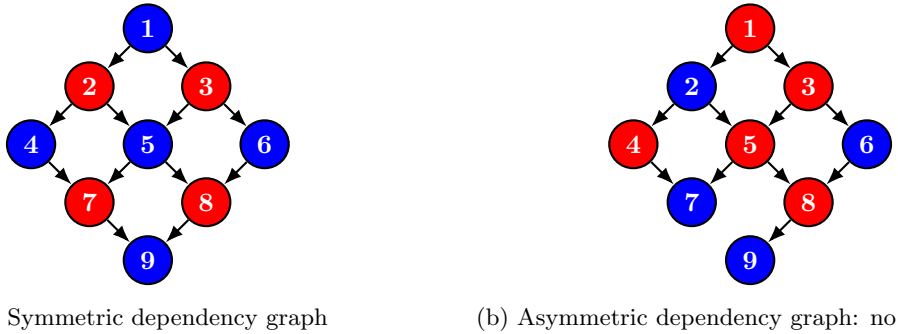
\begin{figure}[H]
    \centering
    \begin{subfigure}[b]{0.48\textwidth}
        \centering
        \begin{tikzpicture}[
            mynode/.style={circle, draw, thick, minimum size=0.2cm, font=\bfseries}, 
            mynodeblue/.style={circle, draw, thick, minimum size=0.2cm, font=\bfseries,fill=blue,text=white}, 
            mynodered/.style={circle, draw, thick, minimum size=0.2cm, font=\bfseries,fill=red,text=white},
            myarrow/.style={-Latex, thick, draw=black},
            bidir_arrow_1to2/.style={myarrow, bend right=20,blue},
            bidir_arrow_2to1/.style={myarrow, bend right=20,blue},
            self_loop_arrow/.style={myarrow, loop right, looseness=10,red}
        ]

            \node[mynodeblue] (N1) {1};
            \node[mynodered, below left=0.3cm and 0.3cm of N1] (N2) {2};
            \node[mynodered, below right=0.3cm and 0.3cm of N1] (N3) {3};
            
            \node[mynodeblue, below left =0.3cm and 0.3cm of N2] (N4) {4};
            \node[mynodeblue, below right =0.3cm and 0.3cm of N3] (N6) {6};
            \node[mynodeblue] (N5) at ($(N4)!0.5!(N6)$) {5};

            \node[mynodered, below right =0.3cm and 0.3cm of N4] (N7) {7};
            \node[mynodered, below left =0.3cm and 0.3cm of N6] (N8) {8};
            \node[mynodeblue, below left =0.3cm and 0.3cm of N8] (N9) {9};

            \draw[myarrow] (N1) to (N2);
            \draw[myarrow] (N1) to (N3);
            \draw[myarrow] (N2) to (N4);
            \draw[myarrow] (N2) to (N5);
            \draw[myarrow] (N3) to (N5);
            \draw[myarrow] (N3) to (N6);
            \draw[myarrow] (N4) to (N7);
            \draw[myarrow] (N5) to (N7);
            \draw[myarrow] (N5) to (N8);
            \draw[myarrow] (N6) to (N8);
            \draw[myarrow] (N7) to (N9);
            \draw[myarrow] (N8) to (N9);

        \end{tikzpicture}
        \caption{Symmetric dependency graph} 
        \label{fig:interdependence_symmetric}
    \end{subfigure}
    \hfill
    \begin{subfigure}[b]{0.48\textwidth}
        \centering
                \begin{tikzpicture}[
            mynode/.style={circle, draw, thick, minimum size=0.2cm, font=\bfseries}, 
            mynodeblue/.style={circle, draw, thick, minimum size=0.2cm, font=\bfseries,fill=blue,text=white}, 
            mynodered/.style={circle, draw, thick, minimum size=0.2cm, font=\bfseries,fill=red,text=white},
            myarrow/.style={-Latex, thick, draw=black},
            bidir_arrow_1to2/.style={myarrow, bend right=20,blue},
            bidir_arrow_2to1/.style={myarrow, bend right=20,blue},
            self_loop_arrow/.style={myarrow, loop right, looseness=10,red}
        ]

            \node[mynodered] (N1) {1};
            \node[mynodeblue, below left=0.3cm and 0.3cm of N1] (N2) {2};
            \node[mynodered, below right=0.3cm and 0.3cm of N1] (N3) {3};
            
            \node[mynodered, below left =0.3cm and 0.3cm of N2] (N4) {4};
            
            \node[mynodeblue, below right =0.3cm and 0.3cm of N3] (N6) {6};
            \node[mynodered] (N5) at ($(N4)!0.5!(N6)$) {5};

            \node[mynodeblue, below right =0.3cm and 0.3cm of N4] (N7) {7};
            \node[mynodered, below left =0.3cm and 0.3cm of N6] (N8) {8};
            \node[mynodeblue, below left =0.3cm and 0.3cm of N8] (N9) {9};

            \draw[myarrow] (N1) to (N2);
            \draw[myarrow] (N1) to (N3);
            \draw[myarrow] (N2) to (N4);
            \draw[myarrow] (N2) to (N5);
            \draw[myarrow] (N3) to (N5);
            \draw[myarrow] (N3) to (N6);
            \draw[myarrow] (N4) to (N7);
            \draw[myarrow] (N5) to (N7);
            \draw[myarrow] (N5) to (N8);
            \draw[myarrow] (N6) to (N8);
            \draw[myarrow] (N8) to (N9);

        \end{tikzpicture}
        \caption{Asymmetric dependency graph: no $7 \to 9$ edge} 
        \label{fig:interdependence_asymmetric}
    \end{subfigure}
    \caption{Example of a symmetric and asymmetric dependency graph (with transition graph identical to their dependency graph) and its asymptotic characterization. The asymmetric graph in (b) is construed by removing the $7 \to 9$ edge. Blue: persistent states. Red: vanishing states.} 
    \label{fig:example_theorem_interdependence}
\end{figure}

In addition, Theorem \ref{thm:vanish_persist_basic} underscores the significant impact that adding or removing an edge in the dependency graph can have on the system's asymptotic characterization. To illustrate this, we modify the system shown in Figure \ref{fig:interdependence_symmetric} by removing the edge from node 7 to node 9. The resulting structure, displayed in Figure \ref{fig:interdependence_asymmetric}, reveals a complete rearrangement of the vanishing and persistent states.

\subsubsection{Complete asymptotic characterization}

An important implication of Theorem~\ref{thm:vanish_persist_basic} is the following: for any interaction-driven dispersion model that satisfies Assumptions \ref{ass:trans_dag} and \ref{ass:depend_dag+}, if, for each cycle, a persistent successor of at least one state in that cycle can be \textit{identified}, then the asymptotic characterization of every state can be determined. To realize this characterization, we show that once the persistent successors in the DAG$^{-}$ dependency graph are identified, one can construct an interaction-driven dispersion model whose dependency graph is a DAG (rather than a DAG$^{-}$) and whose asymptotic characterization matches that of the original model. With this reduction in hand, we then present an iterative procedure that provides a complete asymptotic characterization of all states when the dependency graph is a DAG.

To this end, we state Proposition~\ref{prop:equivalent_model}, which asserts the existence of an equivalent model with the same asymptotic characterization. In this equivalent model, the dependency graph is a DAG rather than the original DAG$^{-}$; the proof of Proposition~\ref{prop:equivalent_model} gives a constructive procedure for building it.

\begin{proposition}\label{prop:equivalent_model}
    Under Assumption \ref{ass:trans_dag} and \ref{ass:depend_dag+}, if, for each cycle, a  persistent successor of at least one state in that  cycle in the dependency graph is identifiable, then there exists an interaction-driven dispersion model with 
    \begin{itemize}
        \item[1.] an identical transition graph,
        \item[2.] and a modified dependency graph that is DAG (instead of DAG$^{-}$), 
    \end{itemize} 
    whose asymptotic characterization is identical to the original model.
\end{proposition}
\begin{proof}
    Our objective is to demonstrate that, under the given conditions, the asymptotic characterization of every state can remain unchanged while the dependency graph is transformed from a DAG$^{-}$ graph to a DAG graph.
    \begin{itemize}
  \item[$\circ$]
\textit{Edge removal process.}
Let $G = (V, E)$ denote the original dependency graph. We provide an iterative method to construct a modified dependency graph $G' = (V, E')$. Choose a cycle $C$ in $G$, we identify a state $s_C \in C$ with a persistent dependency successor $s'_C$, which is possible by the explicit assumption that every state $s_C$ is identifiable. The specific modification of the edge set $E$ depends on the relationship between $s'_C$ and the cycle $C$:

\begin{itemize}
    \item \textit{Case 1: $s'_C \in C$}. In this case, we remove all outgoing dependency edges from state $s'_C$. That is, for every state $X$ such that $(s'_C, X) \in E$, we remove $(s'_C, X)$ from $E$.
    \item \textit{Case 2: $s'_C \notin C$}. In this case, we remove all outgoing dependency edges from state $s_C$, \textit{except} for the edge $(s_C, s'_C)$. That is, for every state $X$ such that $(s_C, X) \in E$ and $X \neq s'_C$, we remove $(s_C, X)$ from $E$.
\end{itemize}
By repeating the above process for each cycle we arrive at a modified edge set $E'$.

  \item[$\circ$]
\textit{Justification of edge removal process.} Let us now justify the impact of the edge removals in each case:

\begin{itemize}
    \item \textit{Justification for Case 1 ($s'_C \in C$):}
    Since $s'_C$ is persistent, by Theorem \ref{thm:vanish_persist_basic}, all of its dependency successors $X$ (if any) must be asymptotically vanishing. Removing all outgoing dependency edges from $s'_C$ does not alter $s'_C$'s persistent status; in fact, it makes $s'_C$ into a sink within the modified graph. Furthermore, the asymptotic characterization of any state $X$ that was a successor of $s'_C$ remains unaffected because, again by Theorem \ref{thm:vanish_persist_basic}, its asymptotic characterization is solely dependent on its own successors. 

    \item \textit{Justification for Case 2 ($s'_C \notin C$):}
    State $s_C$ has $s'_C$ as a persistent dependency successor. By Theorem \ref{thm:vanish_persist_basic}, state $s_C$ must be asymptotically vanishing as long as it retains its dependency on $s'_C$. Hence, removing all other outgoing dependency edges from state $s_C$ (i.e., to successors other than $s'_C$) would not alter the asymptotic characterization of $s_C$. Furthermore, the asymptotic characterization of any state $X$ that was one of these removed successors of $s_C$ remains unaffected because, again by Theorem \ref{thm:vanish_persist_basic}, its asymptotic characterization is solely dependent on its own successors. 
    
\end{itemize}

By the above justification, the proposed edge removal process thus preserves the asymptotic characterization of all states. Note that the conditions for Theorem \ref{thm:vanish_persist_basic} are still valid after each edge removal, which thus allow for this process to continue for each cycle. Crucially, by removing these edges, every cycle in the dependency graph is `broken', which provides us with a modified dependency graph $(V,E')$ that is DAG.
\end{itemize}
This completes the proof. \end{proof}

Now we have modified a DAG$^{-}$ dependency graph into a DAG graph, asymptotically characterizing each state becomes simple. The key insight is that, according to Lemma \ref{lemma:strongly_connected}, every DAG contains a sink. This persistent sink then becomes our starting point for an iterative process of applying Theorem \ref{thm:vanish_persist_basic}  leading to the complete asymptotic characterization of states. We formalize this in the following corollary, of which the proof is provided in Appendix~\ref{app:proofs}.

\begin{corollary}\label{cor:complete_char}
Under Assumptions \ref{ass:trans_dag} and \ref{ass:depend_dag+}, if the persistent successor within each cycle of the dependency graph is identifiable, then the asymptotic characterization of every state can be determined. 

Furthermore, this asymptotic characterization is independent of both the initial state values $($which are in $(0,1)$$)$ and the specific numerical values of the $\beta$ coefficients, depending only on whether each $\beta$-coefficient is zero or non-zero.
\end{corollary}

\subsubsection{Example: asymptotic characterization of SIR and DK models}\label{subsubsect:three_step_sir_dk}
In Section \ref{subsubsect:sir_dk_dag-}, we determined that while both the SIR and DK models exhibit a DAG transition graph, their dependency graphs are DAG$^{-}$ (not DAGs). The cycles within these dependency graphs are illustrated in Figures \ref{fig:dependency_graph_sir_dag+} and \ref{fig:dependency_graph_dk_dag+}. We now present a three-step process to ascertain the asymptotic characterization of both models.
\begin{itemize}
    \item[$\circ$] \textit{Step 1}: Verify that \textit{Assumptions \ref{ass:trans_dag} and \ref{ass:depend_dag+}} are met.
    \item[$\circ$] \textit{Step 2}: If the dependency graph contains any cycles, apply the `edge removal process' as detailed in the proof of Proposition \ref{prop:equivalent_model}.
    \item[$\circ$] \textit{Step 3}: On the (modified) DAG dependency graph, begin at the sink state(s) and iteratively determine the asymptotic characterization of its predecessors using Theorem \ref{thm:vanish_persist_basic}.
\end{itemize}
For both the SIR and DK models, {\it Step 1} has already been covered by Section \ref{subsubsect:sir_dk_dag-}. Regarding {\it Step 2}, we revisit Figure \ref{fig:dependency_graph_sir_dag+}. We observe that the persistent state 3 has predecessor state 2, which is part of the `blue cycle' $C$. In the notation of the edge removal process, this means $s_C = 2$ and $s'_C = 3$. Since $s'_C \notin C$, we are in Case 2 of the edge removal process. Consequently, all outgoing dependency edges from state $s_C = 2$ are removed, with the exception of the edge to $s'_C = 3$. This action specifically removes edge $(2,1)$, thereby breaking the blue cycle. A similar reasoning applies to the `red cycle', leading to the removal of its self-loop. For the DK model, the `green cycle' shown in Figure \ref{fig:dependency_graph_dk_dag+} is handled analogously. The resulting modified dependency graph for both models is thus identical to their respective transition graphs, as depicted in Figure~\ref{fig:transition_graph_sir_dk}. Finally, in {\it Step 3}, we determine that sink state 3 is a persistent state. Applying Theorem~\ref{thm:vanish_persist_basic}, we conclude that state 2 is vanishing. Subsequent application of the same theorem then leads to the conclusion that state 1 is persistent.

While the asymptotic characterization of the SIR and DK models is a well-established result that may not require the full depth of the theory developed here, this framework proves valuable for analyzing more complex structures, as demonstrated in the next example.

\subsubsection{Example: asymptotic characterization of a more complex model}\label{subsubsect:example_chen}

In this section, we apply the three-step procedure outlined in Section~\ref{subsubsect:three_step_sir_dk} to a more intricate model. The model is adapted from the insightful study by Chen and Wang~\cite{Chen2020}, which investigates the effects of rumor credibility, correlation, and crowd classification on the dynamics of rumor spreading. For illustrative purposes, we focus solely on the transition and dependency graphs, omitting the original objectives and substantive findings of the paper.

Consider the model as discussed in~\cite{Chen2020}\footnote{For clarity and brevity of exposition, we have chosen to omit the natural recovery mechanism. States \(1\)–\(5\) correspond to Steady Ignorant, Radical Ignorant, Exposed, Spreader, and Stifler, respectively, with labels as defined in~\cite{Chen2020}.}. Its transition and dependency graphs are shown in Figure \ref{fig:example_chen_wang_2020}.

\begin{figure}[H]
    \centering
    \begin{subfigure}[b]{0.48\textwidth}
        \centering
        \begin{tikzpicture}[
            mynode/.style={circle, draw, thick, minimum size=0.2cm, font=\bfseries}, 
            mynodeblue/.style={circle, draw, thick, minimum size=0.2cm, font=\bfseries,fill=blue,text=white}, 
            mynodered/.style={circle, draw, thick, minimum size=0.2cm, font=\bfseries,fill=red,text=white},
            myarrow/.style={-Latex, thick, draw=black},
            bidir_arrow_1to2/.style={myarrow, bend right=20,blue},
            bidir_arrow_2to1/.style={myarrow, bend right=20,blue},
            self_loop_arrow/.style={myarrow, loop right, looseness=10,red}
        ]

            \node[mynode] (N1) {1};
            \node[mynode, right=0.5cm and 0.5cm of N1] (N2) {2};
            \node[mynode, below =0.5cm and 0.5cm of N1] (N3) {3};
            
            \node[mynode, right =0.5cm and 0.5cm of N3] (N4) {4};          
            \node[mynode, below=0.8cm and 0.8cm of N3] (N5) at ($(N3)!0.5!(N4)$) {5};

            \draw[myarrow] (N1) to (N3);
            \draw[myarrow] (N1) to (N4);
            \draw[myarrow] (N2) to (N4);
            \draw[myarrow] (N3) to (N5);
            \draw[myarrow] (N4) to (N5);

        \end{tikzpicture}
        \caption{Transition graph} 
        \label{fig:transition_chen_wang_2020}
    \end{subfigure}
    \hfill
    \begin{subfigure}[b]{0.48\textwidth}
        \centering
                    \begin{tikzpicture}[
            mynode/.style={circle, draw, thick, minimum size=0.2cm, font=\bfseries}, 
            mynodeblue/.style={circle, draw, thick, minimum size=0.2cm, font=\bfseries,fill=blue,text=white}, 
            mynodered/.style={circle, draw, thick, minimum size=0.2cm, font=\bfseries,fill=red,text=white},
            myarrow/.style={-Latex, thick, draw=black},
            bidir_arrow_1to2/.style={myarrow, bend right=20,blue},
            bidir_arrow_2to1/.style={myarrow, bend right=20,blue},
            self_loop_arrow/.style={myarrow, loop right, looseness=10,red}
        ]

            \node[mynode] (N1) {1};
            \node[mynode, right=0.5cm and 0.5cm of N1] (N2) {2};
            \node[mynode, below =0.5cm and 0.5cm of N1] (N3) {3};
            
            \node[mynode, right =0.5cm and 0.5cm of N3] (N4) {4};          
            \node[mynode, below=0.8cm and 0.8cm of N3] (N5) at ($(N3)!0.5!(N4)$) {5};

            \draw[myarrow] (N1) to (N4);
            \draw[self_loop_arrow] (N4) to (N4);
            \draw[myarrow] (N2) to (N4);
            \draw[myarrow] (N3) to (N5);
            \draw[myarrow] (N4) to (N5);
            \draw[bidir_arrow_1to2] (N3) to (N4);
            \draw[bidir_arrow_2to1] (N4) to (N3);

        \end{tikzpicture}
        \caption{Dependency graph} 
        \label{fig:dependency_chen_wang_2020}
    \end{subfigure}
    \caption{Transition and dependency graph of the model in \cite{Chen2020}, with highlighted cycles in blue and red.} 
    \label{fig:example_chen_wang_2020}
\end{figure}
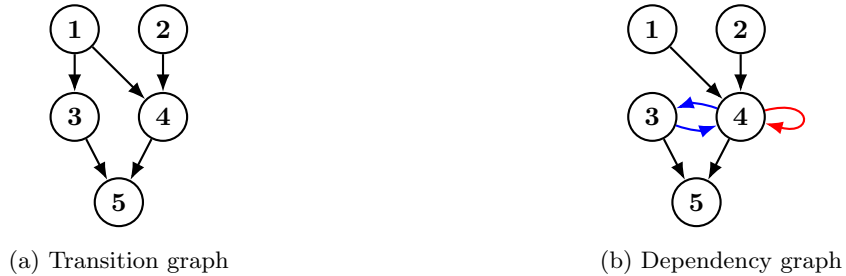

We follow the three-step procedure from Section~\ref{subsubsect:sir_dk_dag-}. For {\it Step~1}, the transition graph in Figure~\ref{fig:transition_chen_wang_2020} is acyclic, so Assumption~\ref{ass:trans_dag} holds. For {\it Step~2}, the dependency graph in Figure~\ref{fig:dependency_chen_wang_2020} contains two cycles (highlighted in red and blue)—the mutual cycle \(3\leftrightarrow4\) and the self-loop \(4\to4\). Both states \(3\) and \(4\) have the sink \(5\) as a persistent dependency successor, hence Assumption~\ref{ass:depend_dag+} is satisfied. Note that \(5\) serves as \(s'_C\) in the `{edge removal process}' from the proof of Proposition~\ref{prop:equivalent_model}. Because \(s'_C=5\) does not lie on either cycle (which involve only \(3\) and \(4\)), we are in Case~2 of that process; accordingly, we remove the edges \(3\to4\), \(4\to3\), and \(4\to4\), yielding a DAG dependency graph as shown in Figure \ref{fig:modified_depend_chen_wang_2020}. Finally, in {\it Step~3} we determine the asymptotic characterization of each state by applying Theorem \ref{thm:vanish_persist_basic} iteratively, starting from the sink state $5$, leading to the asymptotic characterization in Figure \ref{fig:asympt_char_chen_wang_2020}. 
Our conclusions align with the numerical simulations reported in~\cite{Chen2020}, and show that the (Ignorant) states 1 and 2 cannot, by design, be fully converted into (Spreader) states 3 or 4. Furthermore, Theorem~\ref{MR} implies that their proposed model is asymptotically stable --- a property not addressed in~\cite{Chen2020}.

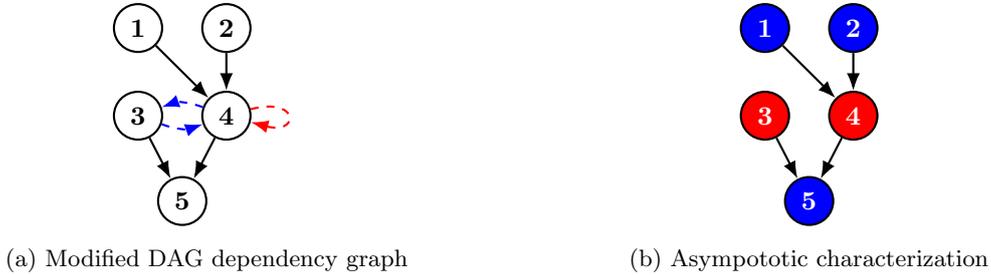
\begin{figure}[H]
    \centering
    \begin{subfigure}[b]{0.48\textwidth}
        \centering
         \begin{tikzpicture}[
            mynode/.style={circle, draw, thick, minimum size=0.2cm, font=\bfseries}, 
            mynodeblue/.style={circle, draw, thick, minimum size=0.2cm, font=\bfseries,fill=blue,text=white}, 
            mynodered/.style={circle, draw, thick, minimum size=0.2cm, font=\bfseries,fill=red,text=white},
            myarrow/.style={-Latex, thick, draw=black},
            bidir_arrow_1to2/.style={myarrow, bend right=20,blue,dashed},
            bidir_arrow_2to1/.style={myarrow, bend right=20,blue,dashed},
            self_loop_arrow/.style={myarrow, loop right, looseness=10,red,dashed}
        ]

            \node[mynode] (N1) {1};
            \node[mynode, right=0.5cm and 0.5cm of N1] (N2) {2};
            \node[mynode, below =0.5cm and 0.5cm of N1] (N3) {3};
            
            \node[mynode, right =0.5cm and 0.5cm of N3] (N4) {4};          
            \node[mynode, below=0.8cm and 0.8cm of N3] (N5) at ($(N3)!0.5!(N4)$) {5};

            \draw[myarrow] (N1) to (N4);
            \draw[self_loop_arrow] (N4) to (N4);
            \draw[myarrow] (N2) to (N4);
            \draw[myarrow] (N3) to (N5);
            \draw[myarrow] (N4) to (N5);
            \draw[bidir_arrow_1to2] (N3) to (N4);
            \draw[bidir_arrow_2to1] (N4) to (N3);

        \end{tikzpicture}
        \caption{Modified DAG dependency graph} 
        \label{fig:modified_depend_chen_wang_2020}
    \end{subfigure}
    \hfill
    \begin{subfigure}[b]{0.48\textwidth}
        \centering
            \begin{tikzpicture}[
            mynode/.style={circle, draw, thick, minimum size=0.2cm, font=\bfseries}, 
            mynodeblue/.style={circle, draw, thick, minimum size=0.2cm, font=\bfseries,fill=blue,text=white}, 
            mynodered/.style={circle, draw, thick, minimum size=0.2cm, font=\bfseries,fill=red,text=white},
            myarrow/.style={-Latex, thick, draw=black},
            bidir_arrow_1to2/.style={myarrow, bend right=20,blue},
            bidir_arrow_2to1/.style={myarrow, bend right=20,blue},
            self_loop_arrow/.style={myarrow, loop right, looseness=10,red}
        ]

            \node[mynodeblue] (N1) {1};
            \node[mynodeblue, right=0.5cm and 0.5cm of N1] (N2) {2};
            \node[mynodered, below =0.5cm and 0.5cm of N1] (N3) {3};
            
            \node[mynodered, right =0.5cm and 0.5cm of N3] (N4) {4};          
            \node[mynodeblue, below=0.8cm and 0.8cm of N3] (N5) at ($(N3)!0.5!(N4)$) {5};

            \draw[myarrow] (N1) to (N4);
            \draw[myarrow] (N2) to (N4);
            \draw[myarrow] (N3) to (N5);
            \draw[myarrow] (N4) to (N5);

        \end{tikzpicture}
        \caption{Asympototic characterization} 
        \label{fig:asympt_char_chen_wang_2020}
    \end{subfigure}
    \caption{Modified dependency graph and asymptotic characterization of the model in Chen \& Wang~\cite{Chen2020}. The dashed edges are removed in order to arrive at the modified dependency graph. The complete characterization then follows from Corollary \ref{cor:complete_char} by applying Theorem \ref{thm:vanish_persist_basic} iteratively starting with the sink state $5$. Blue: asymptotically persistent states. Red: asymptotically vanishing states.} 
    \label{fig:steps_chen_wang_2020}
\end{figure}

\subsection{Vanishing states: exponential versus non-exponential decay}

In this sub section, we examine the decay rates of vanishing states. Our central conjecture is that, under Assumptions~\ref{ass:trans_dag} and~\ref{ass:depend_dag+}, every vanishing state~$i$ decays exponentially; that is, there exist constants $C > 0$ and $r > 0$ such that $y_i(t) \leq C e^{-rt}$. These assumptions, we believe, impose sufficient structure to rule out slower decay modes. However, at present we are unable to prove this conjecture, nor have we succeeded in constructing a counterexample.
To highlight the role of Assumption~\ref{ass:depend_dag+}, we present two examples demonstrating that if this assumption is dropped — while Assumption~\ref{ass:trans_dag} is retained — vanishing states may indeed decay at a subexponential rate.

We begin by emphasizing what the integrability result in Theorem~\ref{th:integrable} does guarantee. Specifically, it ensures that vanishing states under Assumptions \ref{ass:trans_dag} and \ref{ass:depend_dag+} cannot decay at the rate $1/t$. The reason is straightforward: a decay of order $1/t$ is not integrable, and thus would contradict Theorem~\ref{th:integrable}. What the theorem does allow, however, is any decay that is slower than exponential yet still integrable --- for instance, a rate of order $1/t^2$. In this sense, Theorem~\ref{th:integrable} provides a clear lower bound: the decay of vanishing states must be faster than $1/t$.

We now turn to an informal discussion of the heuristic motivation behind our conjecture. By construction of the model, each component of the inflow or outflow can be at most quadratic in the state variables; in particular, higher-order terms (such as cubic or beyond) cannot arise. This quadratic structure is therefore the only plausible source of non-exponential decay, and it naturally suggests a decay rate on the order of $1/t$. In Examples~\ref{ex:self_cycle} and~\ref{example:3-cycle}, we show that such $1/t$-type decay does indeed occur once the assumptions are relaxed. On the other hand, the integrability theorem rules out this behavior under Assumptions \ref{ass:trans_dag} and \ref{ass:depend_dag+}. It is precisely this tension that motivates the conjecture stated below.

\begin{conjecture}\label{conj:exp_decay}
    Under Assumption \ref{ass:trans_dag} and \ref{ass:depend_dag+}, every vanishing state decays exponentially to $0$.
\end{conjecture}

To demonstrate the necessity of Assumption~\ref{ass:depend_dag+} and the occurrence of non-exponential decay, we present two illustrative examples: one where the slower decay arises from a self-cycle, and another where it results from a 3-cycle.

\subsubsection{Example: inversely proportional time decay as result of self-cycle}\label{ex:self_cycle}

Consider two states with $y_1(0)+y_2(0)=1$ and $y_i(0) > 0$ for $i = 1,2$. Let the only nonzero rates be $\beta^{1}_{11}=\beta^{2}_{11}=\alpha>0$. Then the system is
\[
\frac{dy_1}{dt}=-\alpha y_1^2,\qquad \frac{dy_2}{dt}=\alpha y_1^2.
\]
The dependency graph has a self-cycle at $1$, and the transition graph has the single edge $1\to 2$ as depicted in Figure \ref{fig:graphs_selfloop_example}.

\begin{figure}[H]
    \centering
    \begin{subfigure}[b]{0.48\textwidth}
        \centering
        \begin{tikzpicture}[
            mynode/.style={circle, draw, thick, minimum size=0.35cm},
            myarrow/.style={-Latex, thick, draw=black},
            highlight/.style={myarrow, very thick}
        ]
            \node[mynode] (N1) {1};
            \node[mynode, below=0.9cm of N1] (N2) {2};

            \draw[highlight] (N1) -- (N2);
        \end{tikzpicture}
        \caption{Transition graph: single edge $1\to 2$}
        \label{fig:transition_graph_selfloop_example}
    \end{subfigure}\hfill
    \begin{subfigure}[b]{0.48\textwidth}
        \centering
        \begin{tikzpicture}[
            mynode/.style={circle, draw, thick, minimum size=0.35cm},
            myarrow/.style={-Latex, thick, draw=black},
            selfloop/.style={myarrow, loop right, looseness=10, very thick}
        ]
            \node[mynode] (N1) {1};
            \node[mynode, below=0.9cm of N1] (N2) {2};

            \draw[selfloop] (N1) to (N1);
        \end{tikzpicture}
        \caption{Dependency graph: self-cycle at $1$, none at $2$}
        \label{fig:dependency_graph_selfloop_example}
    \end{subfigure}

    \caption{Two-state example of inversely proportional time decay. }
    \label{fig:graphs_selfloop_example}
\end{figure}
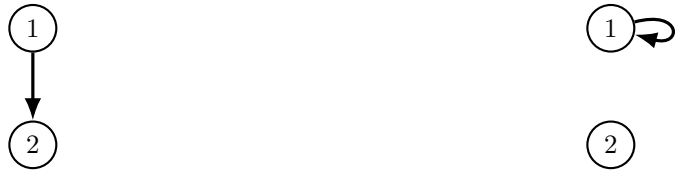
Solving gives
\[
y_1(t)=\frac{y_1(0)}{1+\alpha y_1(0)\,t},\qquad y_2(t)=1-y_1(t).
\]
Thus the vanishing state decays as
\[
y_1(t)\sim \frac{1}{\alpha t},\qquad t\to\infty,
\]
i.e., inversely proportional to time --- rather than exponentially. This confirms that, under Assumption~\ref{ass:trans_dag} but in the absence of Assumption~\ref{ass:depend_dag+}, the presence of a self-cycle can indeed lead to subexponential decay.

\subsubsection{Example: inversely proportional time decay as result a 3-cycle}\label{example:3-cycle}

Consider a four-state model with $\sum_{i=1}^4 y_i(0) = 1$ and $y_i(0) > 0$ for $i = 1, \ldots, 4$. Let the only nonzero rates be $\beta_{11}^1 = \beta_{11}^4 = \alpha_1$, $\beta_{22}^2 = \beta_{22}^4 =  \alpha_2$ and $\beta_{33}^3 = \beta_{33}^4 =  \alpha_3$. The system is then governed by the following set of differential equations:
\begin{equation}
\frac{dy_1}{dt}=-\alpha_1 y_1 y_2,\quad
\frac{dy_2}{dt}=-\alpha_2 y_2 y_3,\quad
\frac{dy_3}{dt}=-\alpha_3 y_3 y_1,\quad
\frac{dy_4}{dt}= \alpha_1 y_1 y_2 + \alpha_2 y_2 y_3 + \alpha_3 y_3 y_1.
\label{eq:three_cycle}
\end{equation}
States $1$, $2$ and $3$ form a 3-cycle in the dependency graph as shown in Figure \ref{fig:graphs_3cycle_example} where both the transition and dependency graphs have been depicted for this system.

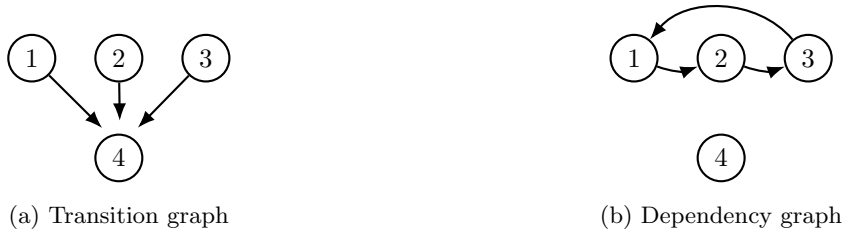
\begin{figure}[H]
    \centering
    \begin{subfigure}[b]{0.48\textwidth}
        \centering
        \begin{tikzpicture}[
            mynode/.style={circle, draw, thick, minimum size=0.35cm},
            myarrow/.style={-Latex, thick, draw=black},
        ]
            \node[mynode] (N1) {1};
            \node[mynode, right=0.5cm of N1] (N2) {2};
            \node[mynode, right=0.5cm of N2] (N3) {3};
            \node[mynode, below=1cm of N2] at (N4) at ($(N1)!0.5!(N3)$) {4};
            
            \draw[myarrow] (N1) to (N4);
            \draw[myarrow] (N2) to (N4);
            \draw[myarrow] (N3) to (N4);
            
        \end{tikzpicture}
        \caption{Transition graph}
        \label{fig:transition_graph_3cycle_example}
    \end{subfigure}\hfill
    \begin{subfigure}[b]{0.48\textwidth}
        \centering
         \begin{tikzpicture}[
            mynode/.style={circle, draw, thick, minimum size=0.35cm},
            myarrow/.style={-Latex, thick, draw=black},
            bidir_arrow_1to2/.style={myarrow, bend right=20},
            bidir_arrow_2to1/.style={myarrow, bend right=50},
        ]
            \node[mynode] (N1) {1};
            \node[mynode, right=0.5cm of N1] (N2) {2};
            \node[mynode, right=0.5cm of N2] (N3) {3};
            \node[mynode, below=1cm of N2] at (N4) at ($(N1)!0.5!(N3)$) {4};
            
            \draw[bidir_arrow_1to2] (N1) to (N2);
            \draw[bidir_arrow_1to2] (N2) to (N3);
            \draw[bidir_arrow_2to1] (N3) to (N1);
            
        \end{tikzpicture}
        \caption{Dependency graph}
        \label{fig:dependency_graph_3cycle_example}
    \end{subfigure}

    \caption{Four-state example of inversely proportional time decay. }
    \label{fig:graphs_3cycle_example}
\end{figure}
Informally, for this system, states $1$, $2$ and $3$ are vanishing and state $4$ is persistent. The decay rate of at least one of the vanishing states is not faster than inversely proportional to time.
Appendix~\ref{app:other} provides formal proofs of the above two statements. These results underscore a key insight: under Assumption~\ref{ass:trans_dag}, but without relying on Assumption~\ref{ass:depend_dag+}, it is not merely self-cycles, but general cyclic structures that can fundamentally induce subexponential decay.

\section{Examples}\label{sec:examples}
In this section, we present examples from epidemiology and rumor propagation, using the theory developed in previous sections to gain insight into the dynamics and limiting behavior of the underlying models.

\subsection{Heterogeneous rumor model}\label{subsubsect:hetero_rumour}
In contrast to epidemiology, where heterogeneous models have been studied more systematically \cite{Diekmann1990,Bootsma2024}, heterogeneous rumor models have received comparatively less attention. Notable contributions include \cite{Nekovee2007, Zhu2020, Zhu2021}. In this subsection, we analyze a heterogeneous rumor model inspired by its epidemic analogue \cite{Chan2013DemographicalChange, Diekmann2013}, focusing on (i) model interpretation, (ii) the asymptotic behavior of states, and (iii) the sensitivity of vanishing-state decay rates to small perturbations in the graph structure.

We consider a population partitioned into two groups: \textit{talk-a-bit} (states 1, 2 and 3) and \textit{talk-a-lot} (states 4, 5 and 6). As the labels suggest, members of the \textit{talk-a-lot} group have a much higher per-contact probability of discussing (and thereby propagating) the rumor than those in the \textit{talk-a-bit} group.

For brevity, we present only the transition and dependency graphs in this section; the full matrix formulation is provided in Appendix~\ref{app:Bmatrix}. As indicated by the transition graph in Figure~\ref{fig:transition_hetero_rumour}, states~\(1\) and~\(4\) represent Ignorants,~\(2\) and~\(5\) are Spreaders, and~\(3\) and~\(6\) are Stiflers.

\begin{figure}[H]
    \centering
    \begin{subfigure}[b]{0.32\textwidth} 
        \centering
        \begin{tikzpicture}[
            mynode/.style={circle, draw, thick, minimum size=0.35cm},
            myarrow/.style={-Latex, thick, draw=black},
            bidir_arrow_1to2/.style={myarrow, bend right=20},
            bidir_arrow_2to1/.style={myarrow, bend right=20},
            self_loop_arrow/.style={myarrow, loop right, looseness=12}            
        ]
            \node[mynode] (N1) {1};
            \node[mynode, below=0.5cm of N1] (N2) {2};
            \node[mynode, below=0.5cm of N2] (N3) {3};

            \node[mynode, right = 0.7cm of N1] (N4) {4};
            \node[mynode, below=0.5cm of N4] (N5) {5};
            \node[mynode, below=0.5cm of N5] (N6) {6};

            \draw[myarrow] (N1) -- (N2); 
            \draw[myarrow] (N2) -- (N3); 
            \draw[myarrow] (N4) -- (N5); 
            \draw[myarrow] (N5) -- (N6); 
        \end{tikzpicture}
        \caption{Transition graph} 
        \label{fig:transition_hetero_rumour}
    \end{subfigure}
    \hfill
    \begin{subfigure}[b]{0.32\textwidth}
        \centering
        \begin{tikzpicture}[
            mynode/.style={circle, draw, thick, minimum size=0.35cm},
            myarrow/.style={-Latex, thick, draw=black},
            bidir_arrow_1to2/.style={myarrow, bend right=20},
            bidir_arrow_2to1/.style={myarrow, bend right=20},
            self_loop_arrow/.style={myarrow, loop right, looseness=12},
            self_loop_arrow2/.style={myarrow, loop left, looseness=12}
        ]
            \node[mynode] (N1) {1};
            \node[mynode, below=0.5cm of N1] (N2) {2};
            \node[mynode, below=0.5cm of N2] (N3) {3};

            \node[mynode, right = 0.7cm of N1] (N4) {4};
            \node[mynode, below=0.5cm of N4] (N5) {5};
            \node[mynode, below=0.5cm of N5] (N6) {6};

            \draw[myarrow] (N1) to (N2);
            \draw[myarrow] (N1) to (N5);
            \draw[myarrow] (N4) to (N2);
            \draw[myarrow] (N4) to (N5);
            \draw[bidir_arrow_1to2] (N2) to (N5);
            \draw[bidir_arrow_2to1] (N5) to (N2);
            \draw[self_loop_arrow] (N5) to (N5);
            \draw[self_loop_arrow2] (N2) to (N2);
            \draw[myarrow] (N2) to (N3);
            \draw[myarrow] (N5) to (N6);
            \draw[myarrow] (N5) to (N3);
            \draw[myarrow] (N2) to (N6);

        \end{tikzpicture}
        \caption{Original dependency graph} 
        \label{fig:dependency_hetero_dk_orig}
    \end{subfigure}
    \hfill
    \begin{subfigure}[b]{0.32\textwidth}
        \centering
              \begin{tikzpicture}[
            mynode/.style={circle, draw, thick, minimum size=0.35cm},
            myarrow/.style={-Latex, thick, draw=black},
            myarrow2/.style={-Latex, thick, draw=red,dashed},
            bidir_arrow_1to2/.style={myarrow, bend right=20},
            bidir_arrow_2to1/.style={myarrow, bend right=20},
            self_loop_arrow/.style={myarrow, loop right, looseness=12},
            self_loop_arrow2/.style={myarrow, loop left, looseness=12}
        ]
            \node[mynode] (N1) {1};
            \node[mynode, below=0.5cm of N1] (N2) {2};
            \node[mynode, below=0.5cm of N2] (N3) {3};

            \node[mynode, right = 0.7cm of N1] (N4) {4};
            \node[mynode, below=0.5cm of N4] (N5) {5};
            \node[mynode, below=0.5cm of N5] (N6) {6};

            \draw[myarrow] (N1) to (N2);
            \draw[myarrow] (N1) to (N5);
            \draw[myarrow] (N4) to (N2);
            \draw[myarrow] (N4) to (N5);
            \draw[bidir_arrow_1to2] (N2) to (N5);
            \draw[bidir_arrow_2to1] (N5) to (N2);
            \draw[self_loop_arrow] (N5) to (N5);
            \draw[self_loop_arrow2] (N2) to (N2);
            \draw[myarrow] (N2) to (N3);
            \draw[myarrow] (N2) to (N6);
            \draw[myarrow2] (N5) to (N6);
            \draw[myarrow2] (N5) to (N3);

        \end{tikzpicture}
        \caption{Alternative dependency graph} 
        \label{fig:dependency_hetero_dk_alt}
    \end{subfigure}
    \caption{Transition and dependency graph of the heterogeneous rumor model, including the alternative dependency graph where the removed edges are marked in dashed red.} 
    \label{fig:graphs_sir_hetero}
\end{figure}
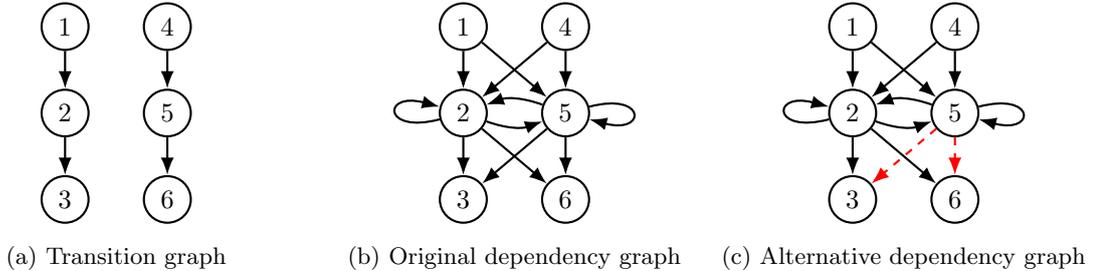

The dependency graph in Figure~\ref{fig:dependency_hetero_dk_orig} is obtained by taking two copies of the DK dependency graph (Figure~\ref{fig:dependency_graph_dk}) — one for each subgroup — and then augmenting them with intergroup dependency edges. 

Let us now go through the three-step process in Section \ref{subsubsect:three_step_sir_dk}. First, we confirm that the transition graph in Figure \ref{fig:transition_hetero_rumour} is a DAG, hence satisfying Assumption \ref{ass:trans_dag}. The cycles in the dependency graph of Figure \ref{fig:dependency_hetero_dk_orig} all have persistent dependency successors state 3 and 6, hence Assumption \ref{ass:depend_dag+} is satisfied as well. By following the edge removal process —  illustrated in Section~\ref{subsubsect:example_chen} — and iteratively applying Theorem \ref{thm:vanish_persist_basic}, we arrive at the asymptotic characterization of the state variables in Figure \ref{fig:asympt_orig_hetero_rumour}.

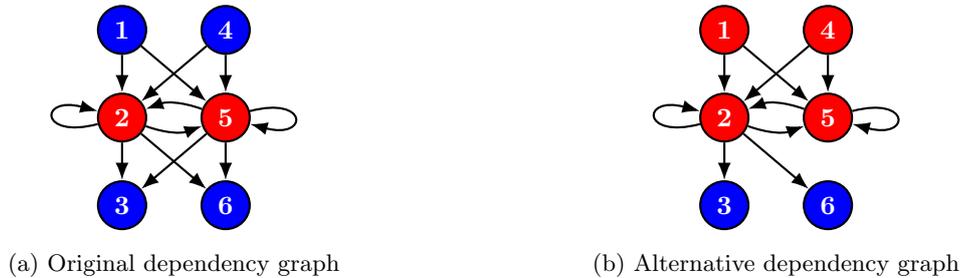
\begin{figure}[H]
    \centering    
    \begin{subfigure}[b]{0.48\textwidth}
        \centering
        \begin{tikzpicture}[
            mynode/.style={circle, draw, thick, minimum size=0.35cm},
            mynodeblue/.style={circle, draw, thick, minimum size=0.2cm, font=\bfseries,fill=blue,text=white}, 
            mynodered/.style={circle, draw, thick, minimum size=0.2cm, font=\bfseries,fill=red,text=white},
            myarrow/.style={-Latex, thick, draw=black},
            bidir_arrow_1to2/.style={myarrow, bend right=20},
            bidir_arrow_2to1/.style={myarrow, bend right=20},
            self_loop_arrow/.style={myarrow, loop right, looseness=12},
            self_loop_arrow2/.style={myarrow, loop left, looseness=12}
        ]
            \node[mynodeblue] (N1) {1};
            \node[mynodered, below=0.5cm of N1] (N2) {2};
            \node[mynodeblue, below=0.5cm of N2] (N3) {3};

            \node[mynodeblue, right = 0.7cm of N1] (N4) {4};
            \node[mynodered, below=0.5cm of N4] (N5) {5};
            \node[mynodeblue, below=0.5cm of N5] (N6) {6};

            \draw[myarrow] (N1) to (N2);
            \draw[myarrow] (N1) to (N5);
            \draw[myarrow] (N4) to (N2);
            \draw[myarrow] (N4) to (N5);
            \draw[bidir_arrow_1to2] (N2) to (N5);
            \draw[bidir_arrow_2to1] (N5) to (N2);
            \draw[self_loop_arrow] (N5) to (N5);
            \draw[self_loop_arrow2] (N2) to (N2);
            \draw[myarrow] (N2) to (N3);
            \draw[myarrow] (N5) to (N6);
            \draw[myarrow] (N5) to (N3);
            \draw[myarrow] (N2) to (N6);

        \end{tikzpicture}
        \caption{Original dependency graph} 
        \label{fig:asympt_orig_hetero_rumour}
    \end{subfigure}
    \hfill
    \begin{subfigure}[b]{0.48\textwidth}
        \centering
              \begin{tikzpicture}[
            mynode/.style={circle, draw, thick, minimum size=0.35cm},
            myarrow/.style={-Latex, thick, draw=black},
            myarrow2/.style={-Latex, thick, draw=red,dashed},
            bidir_arrow_1to2/.style={myarrow, bend right=20},
            bidir_arrow_2to1/.style={myarrow, bend right=20},
            self_loop_arrow/.style={myarrow, loop right, looseness=12},
            mynodeblue/.style={circle, draw, thick, minimum size=0.2cm, font=\bfseries,fill=blue,text=white}, 
            mynodered/.style={circle, draw, thick, minimum size=0.2cm, font=\bfseries,fill=red,text=white},
            self_loop_arrow2/.style={myarrow, loop left, looseness=12}
        ]
            \node[mynodered] (N1) {1};
            \node[mynodered, below=0.5cm of N1] (N2) {2};
            \node[mynodeblue, below=0.5cm of N2] (N3) {3};

            \node[mynodered, right = 0.7cm of N1] (N4) {4};
            \node[mynodered, below=0.5cm of N4] (N5) {5};
            \node[mynodeblue, below=0.5cm of N5] (N6) {6};

            \draw[myarrow] (N1) to (N2);
            \draw[myarrow] (N1) to (N5);
            \draw[myarrow] (N4) to (N2);
            \draw[myarrow] (N4) to (N5);
            \draw[bidir_arrow_1to2] (N2) to (N5);
            \draw[bidir_arrow_2to1] (N5) to (N2);
            \draw[self_loop_arrow] (N5) to (N5);
            \draw[self_loop_arrow2] (N2) to (N2);
            \draw[myarrow] (N2) to (N3);
            \draw[myarrow] (N2) to (N6);

        \end{tikzpicture}
        \caption{Alternative dependency graph} 
        \label{fig:asympt_alt_hetero_rumour}
    \end{subfigure}
    \caption{Dependency graph and asymptotic characterization of states in the original and alternative heterogeneous rumor model. Blue: asymptotic persistent. Red: asymptotic vanishing.} 
    \label{fig:graphs_sir_hetero_alt}
\end{figure}
Motivated by evidence that peer discussion reduces rumor belief \cite{Wang2018GroupRumor}, and that attitude change is more likely under strong argumentation \cite{PettyCacioppo1986ELM}, we adopt the following stylized mechanism: individuals in the \textit{talk-a-lot} group become Stiflers primarily after an intense, two-sided exchange with another active Spreader. Within the model, we encode this by making Spreader--Spreader interaction the sole pathway to stifling in that subgroup.
This assumptions alters the dependency graph, as shown in Figure \ref{fig:dependency_hetero_dk_alt}. Removing the dependency edges $5 \to 3$ and $5 \to 6$ violates Assumption \ref{ass:depend_dag+}, because the self-loop at state 5 does no longer has a persistent dependency successor. A numerical simulation using standard methods then indicates an alternative asymptotic characterization as shown in Figure~\ref{fig:asympt_alt_hetero_rumour}. 

\begin{figure}[H]
    \centering
    \begin{subfigure}[b]{0.45\textwidth}
        \centering
        \includegraphics[width=\textwidth]{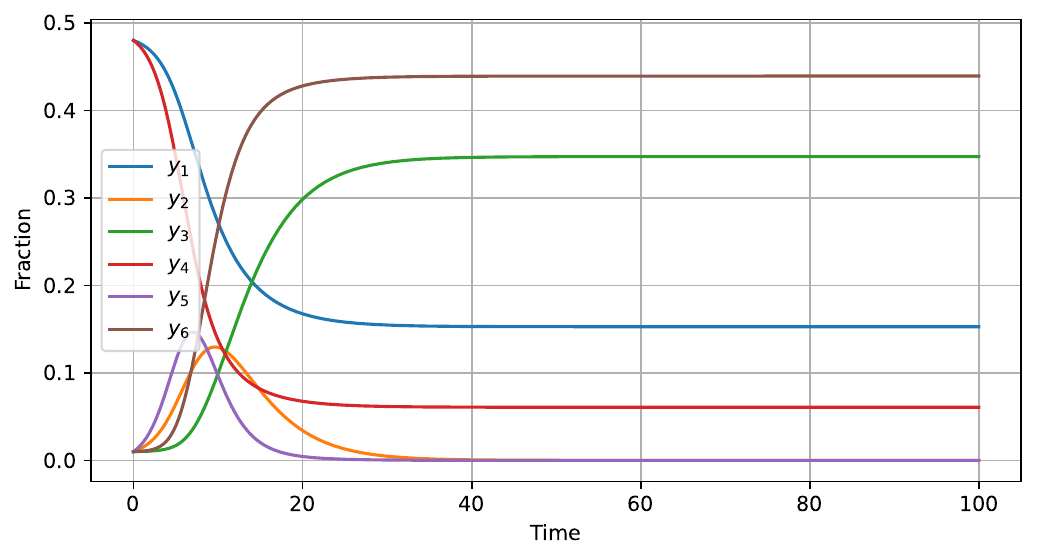}
        \caption{Orignal}
        \label{fig:hetero_dk_norm_all}
    \end{subfigure}
    \hfill
    \begin{subfigure}[b]{0.45\textwidth}
        \centering
        \includegraphics[width=\textwidth]{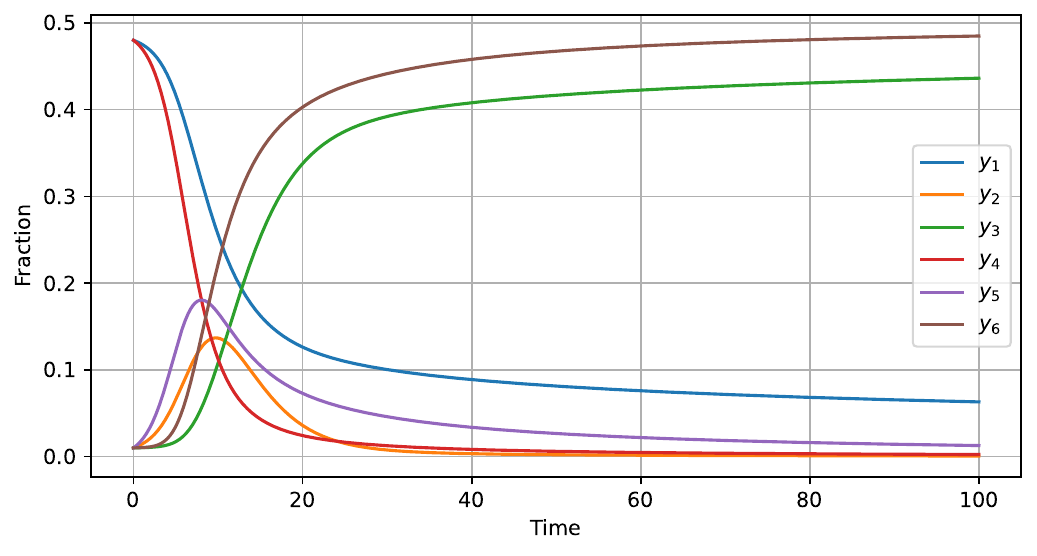}
        \caption{Alternative}
        \label{fig:hetero_dk_alt_all}
    \end{subfigure}
    \caption{Time development of the fractions $y_i$, $i=1, \ldots, 6$ in the heterogeneous rumor model and the alternative version of it. Starting conditions are $y_1(0)=0.48$, $y_2(0)=0.01$, $y_3(0)=0.01$, $y_4(0)=0.48$, $y_5(0)=0.01$, and $y_6(0)=0.01$. The numerical simulation utilized $10000$ steps with a time step of $\Delta t = 0.01$.}
    \label{fig:hetero_dk_all}
\end{figure}
A comparison of the time evolution of the state variables in the original heterogeneous rumor model (Figure~\ref{fig:hetero_dk_norm_all}) and the alternative version (Figure~\ref{fig:hetero_dk_alt_all}) lends support to our conjecture — formulated in Conjecture~\ref{conj:exp_decay} — that subexponential decay may occur when Assumption~\ref{ass:depend_dag+} is not satisfied.

\subsection{SIR-DK model}\label{subsect:sir_dk_hybrid}

Our interaction-driven dispersion framework enables multiple models to be coupled. In this subsection we combine the SIR model with the Daley–Kendall model with two aims: (i) to demonstrate how two canonical models can be integrated; and (ii) to  (numerically) explore various surprising phenomena arising out of information dissemination during an epidemic. Related work includes \cite{Teslya2022}.

Consider a nine-state model in which each state encodes both an SIR compartment and a DK status; the labeling is given in Figure~\ref{fig:state-mapping} (see Appendix~\ref{app:Bmatrix} for the complete model specification).

\begin{figure}[ht]
  \centering
  \begin{tabular}{c|ccc}
    & Ignorant & Spreader & Stifler \\ \hline
    Susceptible & 1 & 2 & 3 \\
    Infectious  & 4 & 5 & 6 \\
    Recovered   & 7 & 8 & 9 \\
  \end{tabular}
  \caption{Interpretation of states in SIR-DK model. Rows: SIR compartments. Columns: DK roles).}
  \label{fig:state-mapping}
\end{figure}

The transition graph in Figure~\ref{fig:transition_SIR-DK} allows moves along the epidemiological (SIR) axis, the rumor (DK) axis, or both simultaneously --- for example, \(1\!\to\!4\) (infection only), \(1\!\to\!2\) (information only), and \(1\!\to\!5\) (both).
\begin{figure}[H]
    \centering
    \begin{subfigure}[t]{0.32\textwidth} 
        \centering
        \begin{tikzpicture}[
            mynode/.style={circle, draw, thick, minimum size=0.35cm},
            myarrow/.style={-{Latex[open]}, thick, draw=black},
            bidir_arrow_1to2/.style={myarrow, bend right=20},
            bidir_arrow_2to1/.style={myarrow, bend right=20},
            self_loop_arrow/.style={myarrow, loop right, looseness=12}            
        ]
            \node[mynode] (N1) {1};
            \node[mynode, right=1.1cm of N1] (N2) {2};
            \node[mynode, right=1.1cm of N2] (N3) {3};

            \node[mynode, below = 0.9cm of N1] (N4) {4};
            \node[mynode, right=1.1cm of N4] (N5) {5};
            \node[mynode, right=1.1cm of N5] (N6) {6};

            \node[mynode, below = 0.9cm of N4] (N7) {7};
            \node[mynode, right=1.1cm of N7] (N8) {8};
            \node[mynode, right=1.1cm of N8] (N9) {9};

            \draw[myarrow] (N1) -- (N2); 
            \draw[myarrow] (N2) -- (N3); 
            \draw[myarrow] (N1) -- (N4); 
            \draw[myarrow] (N1) -- (N5); 
            \draw[myarrow] (N2) -- (N5); 
            \draw[myarrow] (N2) -- (N6); 
            \draw[myarrow] (N3) -- (N6); 
            \draw[myarrow] (N2) -- (N3); 
            \draw[myarrow] (N4) -- (N5); 
            \draw[myarrow] (N5) -- (N6); 
            \draw[myarrow] (N4) -- (N7); 
            \draw[myarrow] (N4) -- (N8); 
            \draw[myarrow] (N5) -- (N8); 
            \draw[myarrow] (N5) -- (N9); 
            \draw[myarrow] (N6) -- (N9); 
            \draw[myarrow] (N7) -- (N8); 
            \draw[myarrow] (N8) -- (N9); 
        \end{tikzpicture}
        \caption{Transition graph} 
        \label{fig:transition_SIR-DK}
    \end{subfigure}
    \hfill
    \begin{subfigure}[t]{0.32\textwidth}
        \centering
        \begin{tikzpicture}[
            mynode/.style={circle, draw, thick, minimum size=0.35cm},
            myarrow/.style={-{Latex[open]}, thick, draw=black},
            qmyarrow/.style={-{Latex[open]}, thick, draw=blue},
            rmyarrow/.style={-{Latex[open]}, thick, draw=red},
            bidir_arrow_1to2/.style={myarrow, bend left=60,draw=red},
            bidir_arrow_2to1/.style={myarrow, bend right=20,draw=red},
            self_loop_arrow/.style={myarrow, loop above, looseness=10,draw=red},
            self_loop_arrow2/.style={myarrow, loop right, looseness=10,draw=red},
            self_loop_arrow3/.style={myarrow, loop below, looseness=10,draw=red},
        ]
            \node[mynode] (N1) {1};
            \node[mynode, right=1.1cm of N1] (N2) {2};
            \node[mynode, right=1.1cm of N2] (N3) {3};

            \node[mynode, below = 0.9cm of N1] (N4) {4};
            \node[mynode, right=1.1cm of N4] (N5) {5};
            \node[mynode, right=1.1cm of N5] (N6) {6};

            \node[mynode, below = 0.9cm of N4] (N7) {7};
            \node[mynode, right=1.1cm of N7] (N8) {8};
            \node[mynode, right=1.1cm of N8] (N9) {9};

            \draw[qmyarrow] (N1) -- (N2); 
            \draw[qmyarrow] (N1) -- (N5); 
            \draw[qmyarrow] (N1) -- (N8); 
            \draw[qmyarrow] (N4) -- (N2); 
            \draw[qmyarrow] (N4) -- (N5); 
            \draw[qmyarrow] (N4) -- (N8); 
            \draw[qmyarrow] (N7) -- (N2); 
            \draw[qmyarrow] (N7) -- (N5); 
            \draw[qmyarrow] (N7) -- (N8); 
            
            \draw[bidir_arrow_2to1] (N2) to (N5); 
            \draw[rmyarrow] (N2) -- (N3); 
            \draw[rmyarrow] (N2) -- (N6); 
            \draw[rmyarrow] (N2) -- (N9); 
            \draw[bidir_arrow_1to2] (N2) to (N8); 
            \draw[self_loop_arrow] (N2) to (N2);

            \draw[bidir_arrow_2to1] (N5) to (N2); 
            \draw[bidir_arrow_2to1] (N5) to (N8); 
            \draw[rmyarrow] (N5) -- (N3); 
            \draw[rmyarrow] (N5) -- (N6); 
            \draw[rmyarrow] (N5) -- (N9); 
            \draw[self_loop_arrow2] (N5) to (N5);
            
            \draw[bidir_arrow_1to2] (N8) to (N2); 
            \draw[bidir_arrow_2to1] (N8) to (N5); 
            \draw[rmyarrow] (N8) -- (N3); 
            \draw[rmyarrow] (N8) -- (N6); 
            \draw[rmyarrow] (N8) -- (N9); 
            \draw[self_loop_arrow3] (N8) to (N8);          
        \end{tikzpicture}
        \caption{Dependency graph of DK-side} 
        \label{fig:dependency_SIR-DK-1}
    \end{subfigure}
    \hfill
    \begin{subfigure}[t]{0.32\textwidth}
        \centering
          \begin{tikzpicture}[
            mynode/.style={circle, draw, thick, minimum size=0.35cm},
            mynode2/.style={circle, draw, thick, minimum size=0.35cm, font=\tiny},
            myarrow/.style={-{Latex[open]}, thick, draw=black},
            p1myarrow/.style={-{Latex[open]}, thick, draw=green},
            p2myarrow/.style={-{Latex[open]}, thick, draw=purple},
            p3myarrow/.style={-{Latex[open]}, thick, draw=cyan},
            alphamyarrow/.style={-{Latex[open]}, thick, draw=orange},
            bidir_arrow_1to2/.style={myarrow, bend left=60},
            bidir_arrow_2to1/.style={myarrow, bend right=20},
            self_loop_arrow/.style={myarrow, loop above, looseness=10},
            self_loop_arrow2/.style={myarrow, loop right, looseness=10},
            self_loop_arrow3/.style={myarrow, loop below, looseness=10},
        ]
            \node[mynode] (N1) {1};
            \node[mynode, right=1.1cm of N1] (N2) {2};
            \node[mynode, right=1.1cm of N2] (N3) {3};

            \node[mynode, below = 0.9cm of N1] (N4) {4};
            \node[mynode, right=1.1cm of N4] (N5) {5};
            \node[mynode, right=1.1cm of N5] (N6) {6};

            \node[mynode, below = 0.9cm of N4] (N7) {7};
            \node[mynode, right=1.1cm of N7] (N8) {8};
            \node[mynode, right=1.1cm of N8] (N9) {9};

            \node[mynode2, right=0.2cm of N7] (N10) {All};

            \draw[p1myarrow] (N1) -- (N4); 
            \draw[p1myarrow] (N1) -- (N5); 
            \draw[p1myarrow] (N1) -- (N6); 
            \draw[p2myarrow] (N2) -- (N4); 
            \draw[p2myarrow] (N2) -- (N5); 
            \draw[p2myarrow] (N2) -- (N6); 
            \draw[p3myarrow] (N3) -- (N4); 
            \draw[p3myarrow] (N3) -- (N5); 
            \draw[p3myarrow] (N3) -- (N6); 

            \draw[alphamyarrow] (N4) -- (N10); 
            \draw[alphamyarrow] (N5) -- (N10); 
            \draw[alphamyarrow] (N6) -- (N10); 
        \end{tikzpicture}
        \caption{Dependency graph of SIR-side} 
        \label{fig:dependency_SIR-DK-2}
    \end{subfigure}
    \caption{Transition and dependency graph of the SIR-DK model. The dependency graph is separated in dependency edges that describe the rumor process (Figure \ref{fig:dependency_SIR-DK-1}) and the epidemiological process (Figure \ref{fig:dependency_SIR-DK-2}). The set of edges of the dependency graph is the union of all the edges in Figures \ref{fig:dependency_SIR-DK-1} and \ref{fig:dependency_SIR-DK-2}. In the dependency graphs, each edge-color denote dependence on a different parameter in the model. The model thus contains 6 parameters in total.}. 
    \label{fig:graphs_sir-dk}
\end{figure}
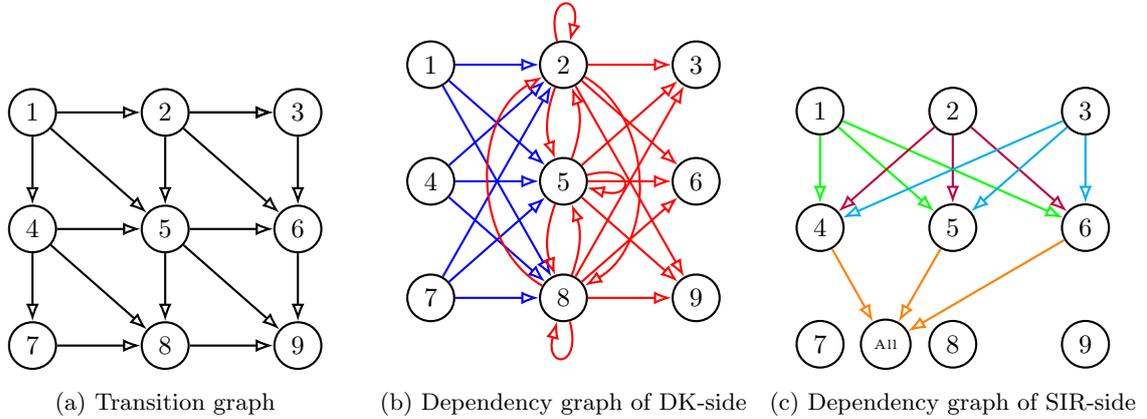
For our illustration, we keep the DK component in its classical two-parameter fluid-limit form. We introduce \(q,r\in[0,1]\): \(q\) is the per-contact rate that an ignorant individual becomes a spreader after meeting a spreader, and \(r\) is the per-contact rate that a spreader becomes a stifler after meeting either a spreader or a stifler. The corresponding dependency graph is shown in Figure~\ref{fig:dependency_SIR-DK-1}, where blue denotes \(q\)-dependent edges and red denotes \(r\)-dependent edges.

Our goal in this short example is to highlight surprising effects induced by the rumor process. To this end on the SIR-side, we let the per-contact infection rate for a susceptible depend on the individual’s DK role: we introduce parameters \(p_i\in[0,1]\) for \(i=1,2,3\), corresponding to the \textit{Ignorant}, \textit{Spreader}, and \textit{Stifler} roles, respectively (see Figure~\ref{fig:dependency_SIR-DK-2} for the SIR-side dependency graph). Recovery follows the classical SIR specification, with per-contact recovery rate \(\alpha\in[0,1]\). 
For simultaneous transition, we adopt the following rule given a contact, where $x$ denotes the transition rate in the DK-process and $y$ the transition rate in the epidemiological process: the transition occurs in both processes at rate $xy$; only in the SIR process at rate $(1 - x) y$; only in the DK process at rate $x (1 - y)$; and no transition occurs in either process at rate $(1 - x)(1 - y)$.

For brevity, we display only the modified dependency graph (a DAG after edge removal; see Proposition \ref{prop:equivalent_model}) together with the resulting asymptotic state classification in Figure \ref{fig:mod_graph_sir-dk}. The intermediate steps follow the template of the preceding sections, and we therefore omit them; the reader may verify Assumptions~\ref{ass:trans_dag} and~\ref{ass:depend_dag+} directly from the graphs.

\begin{figure}[H]
    \centering
        \begin{tikzpicture}[
            mynode/.style={circle, draw, thick, minimum size=0.35cm},
            mynodeblue/.style={circle, draw, thick, minimum size=0.2cm, font=\bfseries,fill=blue,text=white}, 
            mynodered/.style={circle, draw, thick, minimum size=0.2cm, font=\bfseries,fill=red,text=white},
            myarrow/.style={-{Latex[open]}, thick, draw=black},
            bidir_arrow_1to2/.style={myarrow, bend right=20},
            bidir_arrow_2to1/.style={myarrow, bend right=20},
            self_loop_arrow/.style={myarrow, loop right, looseness=12}            
        ]
            \node[mynodeblue] (N1) {1};
            \node[mynodered, right=1.1cm of N1] (N2) {2};
            \node[mynodeblue, right=1.1cm of N2] (N3) {3};

            \node[mynodered, below = 0.9cm of N1] (N4) {4};
            \node[mynodered, right=1.1cm of N4] (N5) {5};
            \node[mynodered, right=1.1cm of N5] (N6) {6};

            \node[mynodeblue, below = 0.9cm of N4] (N7) {7};
            \node[mynodered, right=1.1cm of N7] (N8) {8};
            \node[mynodeblue, right=1.1cm of N8] (N9) {9};

            \draw[myarrow] (N2) -- (N9); 
            \draw[myarrow] (N5) -- (N9); 
            \draw[myarrow] (N4) -- (N9); 
            \draw[myarrow] (N6) -- (N9); 
            \draw[myarrow] (N8) -- (N9); 

            \draw[myarrow] (N1) -- (N2); 
            \draw[myarrow] (N1) -- (N5); 
            \draw[myarrow] (N1) -- (N8); 
            \draw[myarrow] (N1) -- (N4); 
            \draw[myarrow] (N1) -- (N6); 

            \draw[myarrow] (N3) -- (N5); 
            \draw[myarrow] (N3) -- (N4); 
            \draw[myarrow] (N3) -- (N6); 

            \draw[myarrow] (N7) -- (N5); 
            \draw[myarrow] (N7) -- (N2); 
            \draw[myarrow] (N7) -- (N8); 
                        
        \end{tikzpicture}
    \caption{Modified dependency graph of SIR-DK model and the asymptotic characterization.}. 
    \label{fig:mod_graph_sir-dk}
\end{figure}
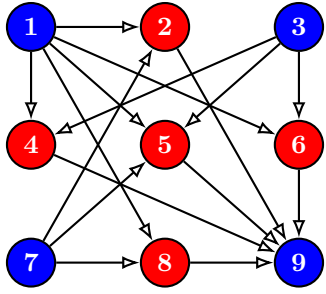

We now numerically investigate surprising behavior that arises when a rumor process is coupled with SIR dynamics. In our first analysis, we assume that individuals in the spreader role have a lower per-contact \textit{infection rate} — due, for example, to increased protective measures — compared to those in the ignorant or stifler roles. This creates dynamics where information dissemination temporarily suppresses disease spread \cite{Funk2009AwarenessPNAS,Wang2015CoupledReview}. Under this assumption, the epidemic may exhibit multi-peak behavior: a secondary wave can emerge, with the notable feature that subsequent peaks may exceed the magnitude of the initial outbreak. This phenomenon aligns with previous findings in coupled behavior–disease models \cite{Epstein2021TripleContagion,Iwasa2023CoupledJTB,Ochab2022Multiwave} as well as empirical analyses of multi-wave epidemics \cite{Cacciapaglia2021Multiwave}.  See Figure~\ref{fig:hybrid_rumour_multi_peak}.

\begin{figure}[H]
    \centering
    \begin{subfigure}[b]{0.45\textwidth}
        \centering
        \includegraphics[width=\textwidth]{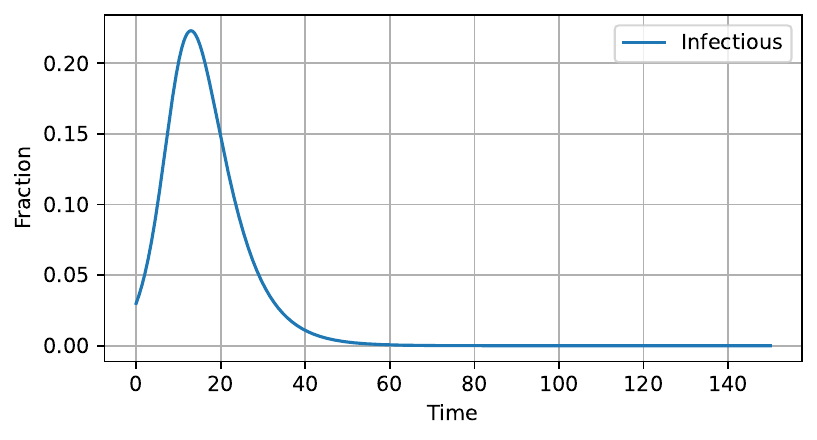}
        \caption{No information dissemination }
        \label{fig:hybrid_no_rumour_inf}
    \end{subfigure}
    \hfill
    \begin{subfigure}[b]{0.45\textwidth}
        \centering
        \includegraphics[width=\textwidth]{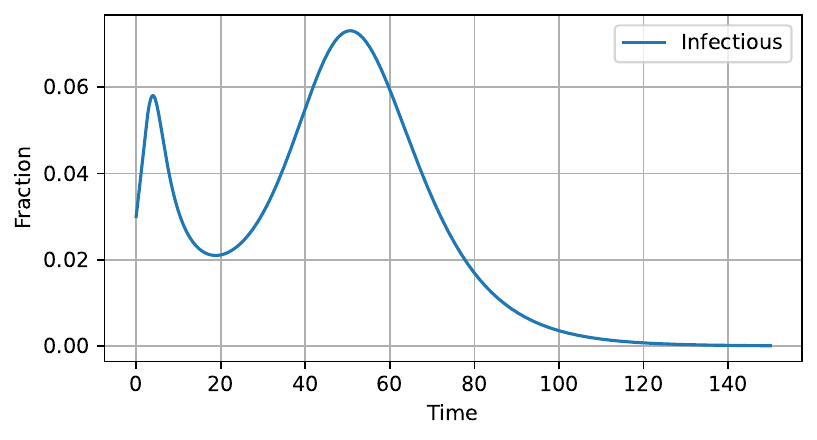}
        \caption{With information dissemination}
        \label{fig:hybrid_high_rumour_inf}
    \end{subfigure}
    \caption{Time development of Infectious  states (i.e., $y_4 + y_5 + y_6$); numerical solution obtained by applying standard numerical integration. Starting conditions are $y_1(0)=0.92$ and $y_2(0) = \ldots = y_9(0) = 0.01$. Chosen parameters are $r = 0.05$, $p_1 = 0.5$, $p_2 = 0.005$, $p_3 = 0.5$, $\alpha = 0.2$. For the scenario `No information dissemination' we used $q = 0$ and for the scenario `With information dissemination' we used $q=1$.}
    \label{fig:hybrid_rumour_multi_peak}
\end{figure}
For our second analysis we studied the effect of the \textit{information} spreading rate $q$ on the total infected fraction. Surprisingly at first glance, our analysis reveals that the fraction of the total infected population may not vary monotonically with the information spreading rate, as demonstrated in Figure \ref{fig:hybrid_phen2}. This phenomenon reflects the cost associated with experiencing two distinct outbreak waves: for low information spreading rate --- circa $q$ less than 0.15 --- the system only shows one wave. Each outbreak with `epidemic momentum' leads to overshoot beyond the herd immunity threshold even after it has been reached \cite{Nguyen2024OvershootNonlinear}. This finding challenges the conventional thinking that maximizing information transparency and dissemination is always optimal for outbreak control \cite{WHO_RCCE_GuidancePortal}.

\begin{figure}[H]
    \centering
    \begin{subfigure}[b]{0.45\textwidth}
        \centering
        \includegraphics[width=\textwidth]{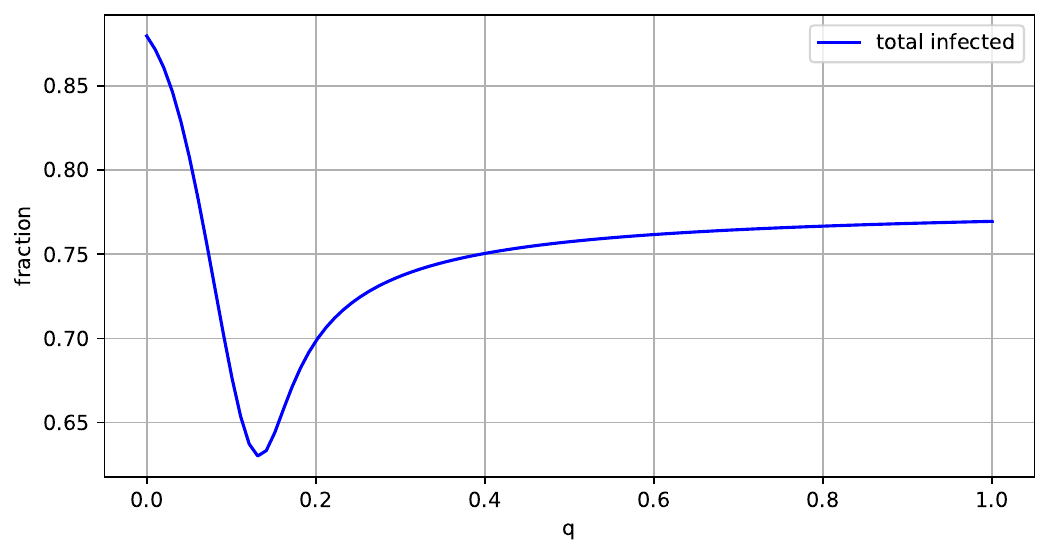}
        \caption{Impact of rumor spreading rate}
        \label{fig:hybrid_phen2}
    \end{subfigure}
    \hfill
    \begin{subfigure}[b]{0.45\textwidth}
        \centering
        \includegraphics[width=\textwidth]{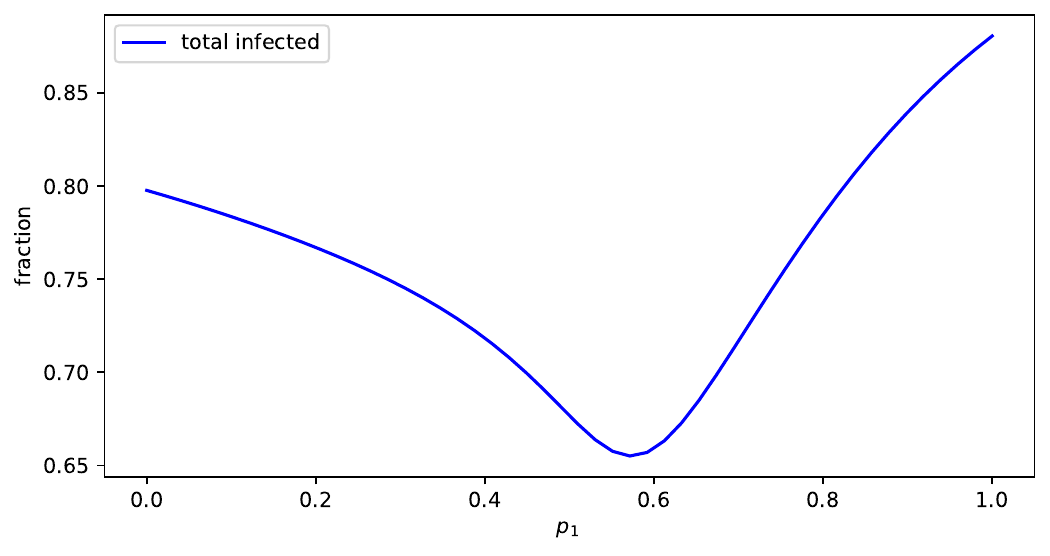}
        \caption{Impact of ignorant infection rate}
        \label{fig:hybrid_phen3}
    \end{subfigure}
    \caption{Development of total infected fraction (sum of the approximated asymptotic values of $y_7, y_8,y_9$) as function of $q$ and $p_1$, the information spreading rate. For each $q$ a numerical solution was found with the Euler method (15,000 iteration steps of timesteps of 0.02). Starting conditions are $y_1(0)=0.92$ and $y_2(0) = \ldots = y_9(0) = 0.01$. Chosen parameters are $r = 0.025$, $p_2 = 0.0025$, $p_3 = 0.25$, $\alpha = 0.1$. In Figure \ref{fig:hybrid_phen2} we used $p_1=0.25$ and in Figure \ref{fig:hybrid_phen3} we used $q=0.5$.
    \label{fig:hybrid_phen23}}
\end{figure}
In a similar vein, we examined how the ignorant \textit{infection} rate \(p_1\) affects the total infection fraction. Contrary to intuition, the relationship is non-monotone: varying \(p_1\) can either increase or decrease the total epidemic size, see Figure~\ref{fig:hybrid_phen3}. These numerical results underscore that coupling two canonical models --- under Assumptions \ref{ass:trans_dag} and \ref{ass:depend_dag+} ---  can generate substantially richer behavior than either model alone. Our framework --- via the graph representations --- provides practical tools to analyze them and extract qualitative insights.

\section{Discussion and concluding remarks}\label{sec:discussion}

In this paper we have introduced a fluid-limit, interaction-driven framework that unifies the SIR and Daley-kendall models, along with many of their extensions, thereby contributing to the literature on epidemic and rumor modeling. Furthermore, for this broad class, we have established theoretical results: under DAG and DAG$^-$ conditions on the transition and dependency graphs, respectively, we derived a classification of asymptotic behavior and formulated a conjecture on exponential decay for vanishing states. Notably, when these structural conditions are violated, the classification can break down, and decay may become non-exponential (e.g., algebraic). On the technical side, we proved asymptotic stability under a DAG transition graph, and $L^1$-integrability of vanishing states under an additional DAG$^-$ condition on the dependency graph. Taken together, these results provide a rigorous foundation for analyzing epidemic and rumor contagion — both as stand-alone processes and in coupled settings — and move beyond model-specific analyses toward general structural insight.

\bigskip

We conclude this paper by an informal discussion of our assumptions and findings, and then  outline directions for future work.

We begin by discussing two standard assumptions in this class of models: homogeneous (uniform) mixing and deterministic (fluid-limit) dynamics. Under homogeneous mixing, every individual is equally likely to contact any other, a simplification of real contact processes. Although a broad class of contact patterns such as  homophily \cite{Mark2003HomophilyDistancing} can be approximated via inflow/outflow parameters and further structure introduced by adding compartments/states, this assumption still excludes richer structures --- such as community structure, and temporal contact patterns --- that can qualitatively change dynamics and produce nonstandard phenomena \cite{Nekovee2007}. 
Our approach is also deterministic: we work in a large-population (fluid-limit) regime in which stochastic fluctuations are negligible. When a compartment becomes small, however, there remains a nonzero probability of stochastic extinction --- initial spreaders may recover before transmitting to anyone who sustains the process --- even when parameters would permit an outbreak. This intrinsically probabilistic phenomenon is not captured by our deterministic formulation. Although these assumptions are standard in theoretical work, we emphasize them here and note that relaxing them, toward finite-population stochastic models or explicit network/temporal contact structure, could enhance realism and test the robustness of our results.

An important modeling choice concerns the initial conditions. We assume that each state variable starts  within the {\it open} interval $(0,1)$, i.e., $0 < y_i(0) < 1$ for all $i$, reflecting the idea that the pathogen or rumor considered has already reached a sufficiently broad set of agent types. This assumption avoids technical boundary issues and ensures that the system remains in the interior of the state space, simplifying the analysis. That said, it would be useful in future work to examine whether the main conclusions continue to hold when this assumption is relaxed — for instance, by allowing some state variables to start at $0$ or $1$, thereby modeling introduction or extinction events more directly.

A central open question is captured by Conjecture~\ref{conj:exp_decay}, which asserts that, under Assumptions~\ref{ass:trans_dag} and~\ref{ass:depend_dag+} (i.e., DAG and DAG$^-$ conditions on the transition and dependency graphs), all vanishing states decay exponentially. The available theoretical evidence and numerical illustrations strongly support this behavior. Yet, every proof attempt we have pursued encounters the same obstacle in a particular edge case — suggesting that a deeper structural insight is still needed to complete the argument. We highlight this as a promising open problem, and would welcome any rigorous resolution.

Finally, we emphasize the scope and generality of the proposed interaction-driven framework, as well as the structural results it yields. Our framework subsumes a broad family of spreading models — including classical examples such as the SIR and Daley-Kendall (DK) models and many of their extensions — and allows for general transition and dependency graphs. In particular, when these graphs coincide and are DAGs, the resulting dynamics give rise to a clear alternating asymptotic pattern of vanishing and persistent states, as illustrated in the symmetric example of Figure~\ref{fig:example_theorem_interdependence}. Importantly, although we assume constant coefficients $\beta$ in the presentation, the theoretical results rely only on the \emph{support} of $\beta$ — that is, whether an entry is identically zero or strictly positive — not on the precise magnitudes. As a result, the conclusions remain valid for time-dependent coefficients $\beta(t)$, provided the zero/nonzero structure remains fixed over time. This significantly broadens the applicability of our results, allowing for models with time-varying spreading or infection rates (e.g., \cite{Lin2020}). In sum, our analysis offers structural insights that go beyond model-specific case studies, highlighting how graph-theoretic properties shape long-term behavior across a wide class of spreading processes.

{\small
\bibliographystyle{ieeetr}
\bibliography{rumour}}

\appendix 

\section{Proofs}\label{app:proofs}

\begin{restatement_lem}[Lemma \ref{lemma:strongly_connected} in main text]
    A directed graph $(V,E)$ that is a DAG contains at least one sink and one source.
\end{restatement_lem}

\begin{proof}
    We prove the following equivalent statement: a directed graph $(V,E)$ with no sinks or sources contains at least one cycle, hence not being a DAG. We focus here on the case with no sinks, but a similar reasoning can be applied on the case with no sources. 
    
    Consider vertex $v$ of a directed graph $(V,E)$ with no sinks. Since vertex $v$ is not a sink it must have at least one outgoing edge. Choose an outgoing edge ${(v,u)}$. Now, since $u$ is also not a sink, $u$ must also have an outgoing edge. If this outgoing edge is ${(u,v)}$, then the subgraph consisting of vertices $v$ and $u$ and edges ${(v,u)}$, ${(u,v)}$ is a cycle. Now, since the directed graph consists of $N$ vertices, it takes at most $N$ repetitions of the above reasoning before we encounter a vertex with a successor that is encountered before, thus closing a cycle. It follows that the directed graph $(V,E)$ with no sinks contains at least one cycle.

\end{proof}
\begin{restatement_lem}[Lemma \ref{lemma:bounded_y} in main text] 
    Assume $0 < y_j(0) < 1$, $j = 1, \ldots, N$ with $\sum_j y_j(0) = 1$. Let $\bm{y}(t)= (y_1(t), \ldots, y_N(t))^{\top}$ satisfy Equation \eqref{eq:model_notnatural}. Then $0 < y_j(t) < 1, j = 1, \ldots, N$ for any $t \geq 0$. 
\end{restatement_lem}
\begin{proof}
We prove this by contradiction. 
    To this end, assume  the existence of a $t_1$ and $i \in \{1,2, \ldots, N\}$ such that $y_i(t_1) \leq 0$. Then the set $A:= \cup_{j=1}^N A_j$, with $A_j := \{t \geq 0: y_j(t) \leq 0 \}$, is non-empty. Now, with $t_0 := \inf\{t\ge 0: t\in A\}$, there exists a $k \in \{1, \ldots, N\}$ such that $y_k(t_0) = 0$ by continuity of $y_k$. Note that $t_0$ must be larger than $0$ since $y_k(0) > 0$. 

    By definition of $t_0$ we have $y_j(\tau) > 0 $ for $\tau \in [0,t_0)$ for all $j \in \{1,  \ldots, N\}$. Thus, the second term of the right-hand side of Equation (\ref{eq:general_1}) must be larger than $0$:
    \begin{equation*}
        \sum_{\substack{1 \leq n,k \leq N, \\n \ne j}} \beta^j_{nk} \,y_n(\tau)\, y_k(\tau) \geq 0, \text{\, for any  } \tau \in [0,t_0).
    \end{equation*}
    Since $\sum_n y_n(\tau) = 1$, we have by the same reasoning that $y_j(\tau) < 1$ for $\tau \in [0,t_0)$. 
    
    We first treat the case of $\beta^j:=\sum_{1 \leq m \leq N} \beta^j_{jm} > 0$. Thus, regarding the first term of the right-hand side of Equation (\ref{eq:general_1}), we find the lower bound

    \begin{equation*}
        - y_j \sum_{1 \leq m \leq N} \beta^j_{jm} y_m > - y_j \sum_{1 \leq m \leq N} \beta^j_{jm}.
    \end{equation*}
Upon combining the above, we thus have  that
    \begin{align*}
        \frac{dy_k}{dt} &= - y_k \sum_{1 \leq m \leq N} \beta^k_{km} y_m + \sum_{\substack{1 \leq n,z \leq N, \\n \ne z}} \beta^k_{nz} y_n y_z \geq - y_k \sum_{1 \leq m \leq N} \beta^k_{km} y_m + 0  \,>\, - y_k \beta^k.\label{eq:unequal_yk}
    \end{align*}
    Thus, on $[0, t_0)$, we have $y_k > g$, with $g(t)$ the solution of the differential equation $g'(t)=-g(t)\,\beta^k$ and initial value $g(0) = y_k(0) > 0$. 
    
    If we can prove that $g$ is a strictly positive function, then indeed  $y_k > g > 0$ on the interval under consideration. 
    It is elementary to show that $g$ is strictly positive. 
    Solving the differential equation for $g$ with $\beta^k > 0$, we obtain
    \begin{equation*}
        g(t) = \exp{\left(-\beta^k\left(t + \frac{-\log y_k(0)}{\beta^k}\right)\right)}=y_k(0)\exp\left(-\beta^k t\right);
    \end{equation*}
    note that $0 < y_k(0) < 1$, hence $\log y_k(0)$ is well-defined. In the case that $\beta^k = 0$ we evidently have that $g(t) = y_k(0)>0$. 

    We have arrived at a contradiction. By continuity we have $y_k(t_0) > g(t_0) > 0$, but we had concluded before that $y_k(t_0) = 0$ by the definition of $t_0 = \inf(A)$. Thus, there cannot exist an $i \in \{1,\ldots, N\}$ and a $t_1$ such that $y_i(t_1) \leq 0$. 
\end{proof}
\begin{restatement_lem}[Lemma \ref{lem:I_leq_O} in main text]
Let ${\mathscr P}_i$ denote the set of transition predecessors of state $i$. Then, the inflow to state $i$ is less than or equal to the sum of the outflows from its transition predecessors:
\[ I_i(t) \leq \sum_{k \in {\mathscr P}_i} O_k (t). \]
\end{restatement_lem}
\begin{proof}
    From the assumption $\sum_j y_k(t) = 1$ for all $t \geq 0$ we know that any outflow from a state leads to an inflow in different state(s). Stated otherwise, for any $k, j$ we have:
    $$\beta^k_{kj} = \sum_{i \ne k} \beta^i_{kj}.$$
    We denote for a state $i$ the total inflow originating from a transition predecessor $k$ by $$I_{i\leftarrow k}(t) := y_k(t) \sum_{i=1}^N \beta^i_{kj} y_j(t).$$

    This inflow $I_{i \leftarrow k}$ is one of the components that sum up to the entire outflow of state $k$: $O_k(t) = y_k(t) \sum_{j=1}^N \beta^k_{kj}y_j$. We therefore have the critical inequality:
    \begin{equation*}
        I_{i \leftarrow k} (t) \leq O_k(t).
    \end{equation*}
    We conclude by noting that $I_i$ is the sum of the inflow originating from its transition predecessors:
    \[         I_{i}(t) = \sum_{k \in {\mathscr P}_i} I_{p_i \leftarrow k}(t) \le \sum_{k \in {\mathscr P}_i} O_k(t).\]
    This completes the proof. 
\end{proof}

\begin{restatement_lem}[Lemma \ref{lem:I_and_O_finite} in main text]
    Let the transition graph be a DAG. For any state $i \in \{1, \ldots, N\}$, the integrals $\int_0^\infty I_i(t)\, dt$ and $\int_0^\infty O_i(t) \,dt$ are finite.
\end{restatement_lem}
\begin{proof}
    We prove this by induction on the transition graph. By the DAG-assumption of the transition graph, we can order and label the nodes $p_1, \ldots, p_N$ such that state $p_n$ can only have inflow from the states $p_j$ with $j < n$. We are going to prove the statement: for $k \leq m$ the integrals $\int_0^\infty I_{p_k}(t)\, dt$ and $\int_0^\infty O_{p_k}(t)\, dt$ are finite. 
    \begin{itemize}
     \item[$\circ$]
    \textit{Base case.} For $m=1$ we trivially have ${\smash \int_0^\infty I_{p_1}(t)\, dt < \infty}$ since $p_1$ is a source node in the transition graph, hence $I_{p_1}(t) = 0$. Integrating the differential equation (\ref{eq:general_1}) for state $p_1$ over the interval $[0, T)$ and taking $T \to \infty$  yields
$$y_{p_1}(\infty) - y_{p_1}(0) = \int_0^\infty I_{p_1}(t)\,dt - \int_0^\infty O_{p_1}(t)\,dt.$$ The left hand side is clearly finite since $y_{p_1}$ lies between $0$ and $1$ (and $y_{p_1}(\infty)$ is well-defined by Theorem \ref{MR}). Combining this with the result $\int_0^\infty I_{p_1}(t)\, dt < \infty$ we conclude that $\int_0^\infty O_{p_1}(t)\,dt$ must be finite as well. 
    \item[$\circ$]
    \textit{Induction step.} Assume that the statement holds for $m-1$. Our objective is to show that $\int_0^\infty I_{p_m}(t) \,dt < \infty$. By Lemma \ref{lem:I_leq_O}, and taking the integral on both sides over the interval $[0,T)$, and then taking $T \to \infty$, we obtain
    \begin{equation}\label{eq:int_I_1}
        \int_0^\infty I_{p_m}(t) \,dt  \le \sum_{k \in {\mathscr P}(p_m)} \int_0^\infty O_k(t) \,dt.
    \end{equation}

    Since the statement holds for $m-1$, we know that for any $k \in {\mathscr P}(p_m)$, the outgoing integral $\int_0^\infty O_k(t)\, dt$ is finite. This directly implies that the incoming integral for state $p_m$, i.e., $\int_0^\infty I_{p_m}(t)\,dt$, must also be finite by Equation (\ref{eq:int_I_1}). Using an identical line of reasoning as applied in the base case, we can then conclude that the outgoing integral $\int_0^\infty O_{p_m}(t)\,dt$ is finite, thereby completing the induction step for a general $m$.
       
    \end{itemize}
    We have thus proven the claim. 
\end{proof}
\begin{restatement_cor}[Corollary \ref{cor:complete_char} in main text]
Under Assumptions \ref{ass:trans_dag} and \ref{ass:depend_dag+}, if the persistent successor within each cycle of the dependency graph is identifiable, then the asymptotic characterization of every state can be determined. 

Furthermore, this asymptotic characterization is independent of both the initial state values and the specific numerical values of the $\beta$ coefficients, depending only on whether each coefficient is zero or non-zero.
\end{restatement_cor}
\begin{proof}
By Proposition \ref{prop:equivalent_model}, we can find a modified dependency graph that is a DAG and whose asymptotic characterization of every state remains unchanged. Let us consider this dependency graph. We can establish a topological ordering of the states, denoted as $p_1, p_2, \ldots, p_N$, such that if there is a dependency edge from state $p_i$ to state $p_j$, then $i < j$. This means any state $p_k$ only has dependency predecessors that appear earlier in the ordering ($p_i$ where $i < k$) and only has dependency successors that appear later in the ordering ($p_j$ where $j > k$).

We will determine the asymptotic behavior of each state by processing them in reverse topological order, starting from $p_N$ and proceeding down to $p_1$.

\begin{enumerate}
    \item[$\circ$] \textit{Base case.} The state $p_N$ is a sink in the dependency graph by virtue of the topological ordering; that is, there are no states $p_j$ such that $p_N \to p_j$ is a dependency edge. Thus, $p_N$ is definitively persistent.

    \item[$\circ$] \textit{Inductive step.} Assume that for all states $p_j$ where $k < j \le N$, their asymptotic behavior has already been precisely characterized. Now, consider state $p_k$. Any dependency successor of $p_k$ must be a state $p_j$ where $j > k$ due to our chosen topological ordering. Since the asymptotic behavior of all such $p_j$ states is known by our inductive hypothesis, we can directly apply Theorem \ref{thm:vanish_persist_basic} to state $p_k$ to determine the asymptotic characterization of state $p_k$.
\end{enumerate}
By iteratively applying this argument from $p_N$ down to $p_1$, we achieve a complete asymptotic characterization for every state in the system.

A crucial aspect of Theorem \ref{thm:vanish_persist_basic} is that the asymptotic characterization it provides is inherently independent of two factors: the specific initial values of the states (as long as they satisfy $0< y_i(0)<1$) and the precise numerical values of the $\beta$ coefficients. For the coefficients, only their binary state (zero or non-zero) is relevant.
\end{proof}
\section{Technical lemma}\label{app:techn}

\begin{lemma}\label{lemma:limit_combi}
    Consider the functions $y(t)$, $f(t)$ and $g(t)$, defined on the domain $\mathbb{R}_{\geq0}$ and satisfying the differential equation \[\frac{dy}{dt} = f(t) + g(t).\] Assume $0 < y(t) < 1$, $\lim_{t \to \infty} g(t) = 0$, and that $f(t)$ is either a non-negative function or non-positive function. Then \[\lim_{t \to \infty} f(t) = 0\quad\mbox{and}\quad \lim_{t \to \infty}\frac{dy}{dt} =0.\]
\end{lemma}
\begin{proof}
    We provide a proof for $f$ a non-negative function, a proof for $f$ a non-positive function follows a similar reasoning. 
    Hence we consider a function $f$ such that $f(t) \geq 0$ for any $t\in{\mathbb R}_{\ge 0}$. As our goal is to prove the claim by contradiction, we assume that $\liminf_{t \to \infty} f(t) > 0$.
    Then there exists an $\epsilon > 0$ and $n > 0$ such that if $t  \geq n$ then $f(t) > \epsilon$ (at least along a subsequence). Now, since $\lim_{t \to \infty} g(t) = 0$, there exists a $k \geq n > 0 $ such that if $t > k$ then  $g(t) \geq -\tfrac{1}{2}\epsilon$. Thus, for $t > k$,
    \begin{equation}
    	\frac{dy}{dt} = f(t) + g(t) \geq \epsilon - \tfrac{1}{2}\epsilon  = \tfrac{1}{2}\epsilon.
    \end{equation}
    Integrating both sides over the interval $[T,T + \Delta)$ with $T > k$ gives
    \begin{equation}
    	y(T + \Delta) – y(T) \geq \tfrac{1}{2}\epsilon \Delta .
    \end{equation}
    Choosing $\Delta = 2 / \epsilon$, we conclude that $y(T + \Delta) \,–\, y(t) \geq 1$, or equivalently $y(T+\Delta) \geq y(T) + 1$. This contradicts the assumption of $0 < y(t) < 1$ for all $t\ge 0$. Hence, recalling that $f(t) \geq 0$, we conclude $\lim_{t \to \infty} f(t) = 0$. Noting that $\frac{dy}{dt}$ is now the sum of two functions which converge to $0$ as $t\to\infty$, so that we have that $\lim_{t \to \infty} \frac{dy}{dt} = 0$ as well.
\end{proof}

\section{Other proofs}\label{app:other}
\subsection{Proof of the statements in Section \ref{example:3-cycle}}
In this appendix, we prove the statements in Section \ref{example:3-cycle}, formalized as Propositions \ref{prop:three_cycle_lim0} and \ref{prop:three_cycle_algebraic} below. The following lemma, though elementary, is included for completeness, as it plays a role in the proof of Proposition \ref{prop:three_cycle_algebraic}.

\begin{lemma}\label{lem:quadr_ineq}
    For non-negative $x_1, x_2, x_3$ we have the inequality \[x_1 x_2 + x_1 x_3 + x_2 x_3 \; \leq \; \frac{(x_1 + x_2 + x_3)^2}{3}.\]
\end{lemma}
\begin{proof}
    By Cauchy-Schwarz we have $(x_1 + x_2 + x_3)^2 \leq 3\,(x_1^2 + x_2^2 + x_3^2)$, hence 
    \begin{equation}\label{eq:cauchy}
        \frac{(x_1 + x_2 + x_3)^2}{3} \;\leq \; (x_1^2 + x_2^2 + x_3^2).
    \end{equation}
    Since $(x_1 + x_2 + x_3)^2 = x_1^2 + x_2^2 + x_3^2 + 2\,(x_1 x_2 + x_1 x_3 + x_2 x_3)$ we obtain
    \begin{align*}
        2\,(x_1 x_2 + x_1 x_3 + x_2 x_3) & = (x_1 + x_2 + x_3)^2 - (x_1^2 + x_2^2 + x_3^2) \leq \frac{2}{3} (x_1 + x_2 + x_3)^2,
    \end{align*}
    where the inequality follows from (\ref{eq:cauchy}). Dividing by $2$ yields the claimed inequality.
\end{proof}

\begin{proposition}\label{prop:three_cycle_lim0}
    Let $y_1,y_2,y_3,y_4$ be solutions of \,$\eqref{eq:three_cycle}$  with $\alpha_i>0$ and $y_i(0)>0$. Then
    \[\lim_{t\to\infty} y_i(t)=0, \quad i=1,2,3.\]
\end{proposition}
\begin{proof}
    First, note that the transition graph in Figure~\ref{fig:transition_graph_3cycle_example} is a DAG; hence, by Theorem~\ref{MR}, the limits $L_i:=\lim_{t\to\infty}y_i(t)$ of the state variables exist. By Lemma~\ref{lemma:bounded_y} we have $L_i\in[0,1]$.
    
    Suppose, for the sake of contradiction, that $L_1>0$. Then there exists $T>0$ such that $y_1(t) \geq {L_1}/{2}$ for $t\geq T$.    
    Since $y_3(t)>0$ (by Lemma~\ref{lemma:bounded_y}), we may divide by $y_3(t)$ and write
    \[\frac{d}{dt}\log y_3(t)\;=\;\frac{dy_3}{dt}\frac{1}{y_3(t)}\;=\;-\alpha_3 y_1(t) \leq -\alpha_3\frac{L_1}{2}\]
    for $t\ge T$. 
    Integrating from $T$ to $t$ yields
    \begin{equation*}
        y_3(t) \leq y_3(T)\,\exp\Big(-\frac{\alpha_3 L_1}{2}(t-T)\Big).
    \end{equation*}
    In particular, by integrating again from $T$ to $t$ and sending $t \to \infty$,
    \begin{equation}\label{eq:finint}
    \int_T^\infty y_3(s)\,ds \; \leq \;\frac{2\,y_3(T)}{\alpha_3 L_1} \;<\; \infty.
    \end{equation}
    Next, using $\dfrac{d}{dt}\log y_2=-\alpha_2 y_3$ and integrating from $0$ to $t$, we obtain 
    \begin{equation*}
        y_2(t) = y_2(0)\exp\Big(-\alpha_2\int_0^t y_3(s)ds\Big).
    \end{equation*}
    Letting $t\to\infty$, using that \eqref{eq:finint} implies that $\int_0^\infty y_3(s)\, ds<\infty$, and recalling that $y_2(0) > 0$, we obtain the positive limit
    \begin{equation*}
        L_2 = \lim_{t\to\infty}y_2(t) = y_2(0)\exp\Big(-\alpha_2 \int_0^\infty y_3(s)\,ds\Big) > 0.
    \end{equation*}
    Hence there exists $T'>T$ such that $y_2(t) \geq L_2 / 2$, for $t\geq T'$.    
    Returning to the equation governing $y_1$ and dividing by $y_1(t)>0$,
    \[
    \frac{d}{dt}\log y_1(t) = \frac{dy_1}{dt}\frac{1}{y_1(t)} = -\alpha_1 y_2(t) \leq -\alpha_1 \frac{L_2}{2}, \text{ for } t\ge T'.
    \]
    Integrating from $T'$ to $t$ yields
    \[
    0 < y_1(t)\ \le\ y_1(T')\,\exp\!\Big(-\frac{\alpha_1 L_2}{2}\,(t-T')\Big),
    \]
    which entails that $y_1$ has limit $0$ since $L_2 > 0$. This contradicts the initial assumption that $L_1 = \lim_{t \to \infty} y_1(t) > 0$, so that we conclude that $L_1=0$. In view of the system's inherent symmetry, the same reasoning shows $L_2=L_3=0$, proving the claimed statement.
\end{proof}

\begin{proposition}\label{prop:three_cycle_algebraic}
    Let $y_1,y_2,y_3,y_4$ be a solution of \,$\eqref{eq:three_cycle}$ with $\alpha_i>0$ and $y_i(0)>0$. Then there exist constants $c >0$ and $k > 0$ such that 
    \[ \max \{y_1(t), y_2(t), y_3(t)\} \;\geq \; \frac{1}{c \, t + k}, \text{ for } \; t \geq 0.  \]
\end{proposition}
\begin{proof}
    We first derive a lower bound for $S(t):= y_1 + y_2 + y_3$; since $\max \{y_1,y_2,y_3 \} \geq S(t)/3$ this immediately yields the bound up to the factor $1/3$. 
    Differentiating $S(t)$,
    \begin{equation*}
        \frac{dS(t)}{dt} = -\alpha_1 y_1 y_2 -\alpha_2 y_2 y_3 - \alpha_3 y_3 y_1 
         \geq -\alpha_{\max}(y_1 y_2 + y_2 y_3 + y_3 y_1),
    \end{equation*}
    with $\alpha_{\max}:= \max{\{\alpha_1, \alpha_2, \alpha_3\}}$. 
    By Lemma \ref{lem:quadr_ineq},  \[(y_1 y_2 + y_2 y_3 + y_3 y_1) \;\leq\; \frac{(y_1 + y_2 + y_3)^2}{3} \;=\; \frac{S(t)^2}{3},\] 
    hence
    \[\frac{dS(t)}{dt} \geq -\frac{\alpha_{\max}}{3} S(t)^2.\] Since $S(t) > 0$ (by Lemma \ref{lemma:bounded_y} and positive initial values), we may divide by $S(t)^2$ to obtain
    \[-\frac{dS(t)}{dt}\frac{1}{S^2(t)} \leq \frac{\alpha_{\max}}{3}.\] The left-hand side is the derivative of ${1}/{S(t)}$, so integrating both sides from $0$ to $t$ gives the inequality 
    \[\frac{1}{S(t)} - \frac{1}{S(0)} \leq \frac{\alpha_{\max}}{3} t,\] or, equivalently, 
    \[S(t) \geq \frac{1}{{\alpha_{\max}}t/{3}  + {1}/{S(0)}}.\] Finally, noting that $\max{\{y_1,y_2,y_3\}} \geq S(t)/3$, we obtain
    \[\max{\{y_1,y_2,y_3\}} \; \geq \; \frac{1}{\alpha_{\max} \, t + {3}/{S(0)}}, \]
    for all $t \geq 0$. The desired is proven by choosing $c := \alpha_{\max}$ and $k := 3/ S(0)$.    
\end{proof}
\section{Explicit matrix formulations}\label{app:Bmatrix}
In this appendix we detail a number of matrix formulations that have been used in the paper. 

\subsection{Formulation of the model in matrix notation}\label{subsubsect:matrix_notation}

As $N$ (i.e., the number of states) and the complexity of the transition and dependency graphs increase, the system of coupled differential equations defined by (\ref{eq:general_1}) becomes increasingly unwieldy. Fortunately, it can be reformulated compactly in matrix notation, providing a structured representation of the system’s full dynamics. This formulation also facilitates numerical integration of (\ref{eq:general_1}) with explicit methods (e.g. forward Euler or Runge-Kutta), involving only straightforward matrix-vector multiplications.

To express the full dynamics described by (\ref{eq:general_1}) in matrix form, we define the $N \times N$ matrix $B_i$ as follows. Each row $n \ne i$ contains elements $\beta^i_{nk}$, while row $n = i$ contains elements $-\beta^i_{ik}$ (note the minus sign here):
\begin{equation} \label{eq:lambda_matrix}
    B_{i} :=     
    \begin{bmatrix}
        \beta^{i}_{11} & \beta^{i}_{12} & \beta^{i}_{13} & \cdots & \beta^{i}_{1N} \\
        \vdots & \vdots & \vdots & \ddots & \vdots \\
        -\beta^{i}_{i1} & -\beta^{i}_{i2} & -\beta^{i}_{i3} & \cdots & -\beta^{i}_{iN} \\
        
        \vdots & \vdots & \vdots & \ddots & \vdots \\
        \beta^{i}_{N1}  & \beta^{i}_{N2}  & \beta^{i}_{N3}      & \cdots & \beta^{i}_{NN}
    \end{bmatrix}
\end{equation}
with $\beta^i_{nk} \geq 0$. The column vector consisting of all fractions $y_i$ (leaving out the argument $t$ for ease) is denoted by $\bm{y} \equiv (y_1,\ldots,y_n)^{\top}$. We can now write
\begin{equation*}
    \frac{dy_i}{dt} = \bm{y}^{\top} B_{i} \bm{y}.
\end{equation*}
Now, the conservation law as described by (\ref{eq:cons_law}) can be represented by,  
\begin{equation*}\label{eq:closed_B}
    \sum_{1 \leq i \leq N} B_{i} = \mathbf{0}_{N \times N},
\end{equation*}
with $\mathbf{0}_{N \times N}$ denoting the $(N\times N)$ all-zeroes matrix.
Define the $(N^2 \times N^2)$-matrix $B$ as follows: 
\begin{equation*}
    B := 
        \begin{bmatrix}
        B_1 & \mathbf{0}_{N \times N} & \mathbf{0}_{N \times N} & \cdots & \mathbf{0}_{N \times N} \\
        \mathbf{0}_{N \times N}      & B_2 & \mathbf{0}_{N \times N} & \cdots & \mathbf{0}_{N \times N} \\
        \mathbf{0}_{N \times N}      & \mathbf{0}_{N \times N}      & B_3 & \cdots & \mathbf{0}_{N \times N} \\
        \vdots & \vdots & \vdots & \ddots & \vdots \\
        \mathbf{0}_{N \times N}      & \mathbf{0}_{N \times N}      & \mathbf{0}_{N \times N}      & \cdots & B_N,
    \end{bmatrix},
\end{equation*}
the $(N \times N^2)$-matrix $Y$ by
\begin{equation*}
    Y(\bm{y}) := 
    \begin{bmatrix}
        \bm{y}^{\top} & \mathbf{0}_{1 \times N} & \cdots & \mathbf{0}_{1 \times N} \\
        \mathbf{0}_{1 \times N} & \bm{y}^{\top} & \cdots & \mathbf{0}_{1 \times N} \\
        \vdots & \vdots & \ddots & \vdots \\
        \mathbf{0}_{1 \times N} & \mathbf{0}_{1 \times N} & \mathbf{0}_{1 \times N} & \bm{y}^{\top}
    \end{bmatrix},
\end{equation*}
and the $N$ times stacked $\bm{y}$ vector of dimension $N^2 \times 1$ by
\begin{equation*}
    \bm{\hat{y}} := 
    \begin{bmatrix}
        \bm{y}\\
        \vdots\\
        \bm{y}
    \end{bmatrix}.
\end{equation*}
Upon combining the above, we conclude that the system of coupled differential equations of the interaction-driven dispersion model (\ref{eq:general_1}) is then described equivalently by the following system of coupled differential equations:
\begin{equation} \label{eq:model_notnatural}
    \frac{d\bm{y}}{dt} = Y(\bm{y})\, B\, \hat{\bm{y}}.
\end{equation}

\subsection{Matrix formulation of SIR and DK model.}
For the SIR model we have the matrices $B_1$, $B_2$, and $B_3$ given by
\[
B_1 = 
\begin{bmatrix}
0 & -\beta & 0 \\
0 & 0 & 0 \\
0 & 0 & 0
\end{bmatrix}, \quad
B_2 = 
\begin{bmatrix}
0 & \beta & 0 \\
-\gamma & -\gamma & -\gamma \\
0 & 0 & 0
\end{bmatrix}, \quad
B_3 = 
\begin{bmatrix}
0 & 0 & 0 \\
\gamma & \gamma & \gamma \\
0 & 0 & 0
\end{bmatrix},
\]
whereas for the DK, using the notation $B_1^*$, $B_2^*$ and $B_3^*$, 
\[
B_1^* = 
\begin{bmatrix}
0 & -\theta & 0 \\
0 & 0 & 0 \\
0 & 0 & 0
\end{bmatrix}, \quad
B_2^* = 
\begin{bmatrix}
0 & \theta & 0 \\
0 & -\alpha & -\alpha \\
0 & 0 & 0
\end{bmatrix}, \quad
B_3^* = 
\begin{bmatrix}
0 & 0 & 0 \\
0 & \alpha & \alpha \\
0 & 0 & 0
\end{bmatrix}.
\]

\subsection{Matrix formulation of heterogeneous rumor model}

We provide the matrices $B_i$ as defined in (\ref{eq:lambda_matrix}), which completely describe the model dynamics of the heterogeneous rumor model in Section \ref{subsubsect:hetero_rumour}:

\begingroup 
\setlength{\arraycolsep}{3pt} 

\begin{align*}
B_1 &= \begin{bmatrix}
0 & -\frac{1}{2} & 0 & 0 & -\frac{1}{4} & 0 \\
0 & 0 & 0 & 0 & 0 & 0 \\
0 & 0 & 0 & 0 & 0 & 0 \\
0 & 0 & 0 & 0 & 0 & 0 \\
0 & 0 & 0 & 0 & 0 & 0 \\
0 & 0 & 0 & 0 & 0 & 0
\end{bmatrix},
&
B_2 &= \begin{bmatrix}
0 & \frac{1}{2} & 0 & 0 & \frac{1}{4} & 0 \\
0 & -\frac{1}{2} & -\frac{1}{2} & 0 & -\frac{1}{4} & -\frac{1}{4} \\
0 & 0 & 0 & 0 & 0 & 0 \\
0 & 0 & 0 & 0 & 0 & 0 \\
0 & 0 & 0 & 0 & 0 & 0 \\
0 & 0 & 0 & 0 & 0 & 0
\end{bmatrix},
&
B_3 &= \begin{bmatrix}
0 & 0 & 0 & 0 & 0 & 0 \\
0 & \frac{1}{2} & \frac{1}{2} & 0 & \frac{1}{4} & \frac{1}{4} \\
0 & 0 & 0 & 0 & 0 & 0 \\
0 & 0 & 0 & 0 & 0 & 0 \\
0 & 0 & 0 & 0 & 0 & 0 \\
0 & 0 & 0 & 0 & 0 & 0
\end{bmatrix}
\\[1em]
B_4 &= \begin{bmatrix}
0 & 0 & 0 & 0 & 0 & 0 \\
0 & 0 & 0 & 0 & 0 & 0 \\
0 & 0 & 0 & 0 & 0 & 0 \\
0 & -\frac{1}{2} & 0 & 0 & -1 & 0 \\
0 & 0 & 0 & 0 & 0 & 0 \\
0 & 0 & 0 & 0 & 0 & 0
\end{bmatrix},
&
B_5 &= \begin{bmatrix}
0 & 0 & 0 & 0 & 0 & 0 \\
0 & 0 & 0 & 0 & 0 & 0 \\
0 & 0 & 0 & 0 & 0 & 0 \\
0 & \frac{1}{2} & 0 & 0 & 1 & 0 \\
0 & -\frac{1}{2} & -\frac{1}{2} & 0 & -1 & -1 \\
0 & 0 & 0 & 0 & 0 & 0
\end{bmatrix},
&
B_6 &= \begin{bmatrix}
0 & 0 & 0 & 0 & 0 & 0 \\
0 & 0 & 0 & 0 & 0 & 0 \\
0 & 0 & 0 & 0 & 0 & 0 \\
0 & 0 & 0 & 0 & 0 & 0 \\
0 & \frac{1}{2} & \frac{1}{2} & 0 & 1 & 1 \\
0 & 0 & 0 & 0 & 0 & 0
\end{bmatrix}
\end{align*}
\endgroup

\subsection{Matrix formulation of SIR-DK model}

Based on the transmission assumptions in the SIR-DK as described in Section \ref{subsect:example_sir_dk}, we can systematically construct the rate matrices that govern the system dynamics. The complete model specification is captured by the following $B$-matrices defined in \eqref{eq:lambda_matrix}, which encode all possible state transitions:

\tiny 
\setlength{\arraycolsep}{1.5pt}
\setlength{\tabcolsep}{1.5pt}
\renewcommand{\arraystretch}{0.8}

\begin{align*}
B_1 &= \begin{bmatrix}
0 & -q & 0 & -p_1 & \scriptsize\begin{matrix} -q p_1 \\ -(1-q) p_1  \\ -q (1-p_1) \end{matrix} & -p_1 & 0 & -q & 0 \\
0 & 0 & 0 & 0 & 0 & 0 & 0 & 0 & 0 \\
0 & 0 & 0 & 0 & 0 & 0 & 0 & 0 & 0 \\
0 & 0 & 0 & 0 & 0 & 0 & 0 & 0 & 0 \\
0 & 0 & 0 & 0 & 0 & 0 & 0 & 0 & 0 \\
0 & 0 & 0 & 0 & 0 & 0 & 0 & 0 & 0 \\
0 & 0 & 0 & 0 & 0 & 0 & 0 & 0 & 0 \\
0 & 0 & 0 & 0 & 0 & 0 & 0 & 0 & 0 \\ 
0 & 0 & 0 & 0 & 0 & 0 & 0 & 0 & 0 
\end{bmatrix} \\[1ex]
B_2 &= \begin{bmatrix}
0 & q & 0 & 0 & q(1-p_1) & 0 & 0 & q & 0 \\
0 & -r & -r & -p_2 & -\scriptsize\begin{matrix} r p_2 + \\ (1-r) p_2 + \\ r (1-p_2) \end{matrix} & -\scriptsize\begin{matrix} r p_2 + \\ (1-r) p_2 + \\ r (1-p_2) \end{matrix} & 0 & -r & -r \\
0 & 0 & 0 & 0 & 0 & 0 & 0 & 0 & 0 \\
0 & 0 & 0 & 0 & 0 & 0 & 0 & 0 & 0 \\
0 & 0 & 0 & 0 & 0 & 0 & 0 & 0 & 0 \\
0 & 0 & 0 & 0 & 0 & 0 & 0 & 0 & 0 \\
0 & 0 & 0 & 0 & 0 & 0 & 0 & 0 & 0 \\
0 & 0 & 0 & 0 & 0 & 0 & 0 & 0 & 0 \\
0 & 0 & 0 & 0 & 0 & 0 & 0 & 0 & 0 
\end{bmatrix} \\[1ex]
B_3 &= \begin{bmatrix}
0 & 0 & 0 & 0 & 0 & 0 & 0 & 0 & 0 \\
0 & r & r & 0 & r (1-p_2) & r (1-p_2) & 0 & r & r \\
0 & 0 & 0 & -p_3 & -p_3 & -p_3 & 0 & 0 & 0 \\
0 & 0 & 0 & 0 & 0 & 0 & 0 & 0 & 0 \\
0 & 0 & 0 & 0 & 0 & 0 & 0 & 0 & 0 \\
0 & 0 & 0 & 0 & 0 & 0 & 0 & 0 & 0 \\
0 & 0 & 0 & 0 & 0 & 0 & 0 & 0 & 0 \\
0 & 0 & 0 & 0 & 0 & 0 & 0 & 0 & 0 \\
0 & 0 & 0 & 0 & 0 & 0 & 0 & 0 & 0 
\end{bmatrix} \\[1ex]
B_4 &= \begin{bmatrix}
0 & 0 & 0 & p_1 & p_1 (1-q) & p_1 & 0 & 0 & 0 \\
0 & 0 & 0 & 0 & 0 & 0 & 0 & 0 & 0 \\
0 & 0 & 0 & 0 & 0 & 0 & 0 & 0 & 0 \\
-\alpha & -\scriptsize\begin{matrix} \alpha q + \\ (1-\alpha) q + \\ \alpha (1-q) \end{matrix} & -\alpha & -\alpha & -\scriptsize\begin{matrix} \alpha q + \\ (1-\alpha) q + \\ \alpha (1-q) \end{matrix} & -\alpha & -\alpha & -\scriptsize\begin{matrix} \alpha q + \\ (1-\alpha) q + \\ \alpha (1-q) \end{matrix} & -\alpha \\
0 & 0 & 0 & 0 & 0 & 0 & 0 & 0 & 0 \\
0 & 0 & 0 & 0 & 0 & 0 & 0 & 0 & 0 \\
0 & 0 & 0 & 0 & 0 & 0 & 0 & 0 & 0 \\
0 & 0 & 0 & 0 & 0 & 0 & 0 & 0 & 0 \\
0 & 0 & 0 & 0 & 0 & 0 & 0 & 0 & 0 
\end{bmatrix} \\[1ex]
B_5 &= \begin{bmatrix}
0 & 0 & 0 & 0 & q p_1 & 0 & 0 & 0 & 0 \\
0 & 0 & 0 & p_2 & p_2 (1-r) & p_2 (1-r) & 0 & 0 & 0 \\
0 & 0 & 0 & 0 & 0 & 0 & 0 & 0 & 0 \\
0 & q (1-\alpha) & 0 & 0 & q (1-\alpha) & 0 & 0 & q (1-\alpha) & 0 \\
-\alpha & -\scriptsize\begin{matrix} \alpha r + \\ (1-\alpha) r + \\ \alpha (1-r) \end{matrix} & -\scriptsize\begin{matrix} \alpha r + \\ (1-\alpha) r + \\ \alpha (1-r) \end{matrix} & -\alpha & -\scriptsize\begin{matrix} \alpha r + \\ (1-\alpha) r + \\ \alpha (1-r) \end{matrix} & -\scriptsize\begin{matrix} \alpha r + \\ (1-\alpha) r + \\ \alpha (1-r) \end{matrix} & -\alpha & -\scriptsize\begin{matrix} \alpha r + \\ (1-\alpha) r + \\ \alpha (1-r) \end{matrix} & -\scriptsize\begin{matrix} \alpha r + \\ (1-\alpha) r + \\ \alpha (1-r) \end{matrix} \\
0 & 0 & 0 & 0 & 0 & 0 & 0 & 0 & 0 \\
0 & 0 & 0 & 0 & 0 & 0 & 0 & 0 & 0 \\
0 & 0 & 0 & 0 & 0 & 0 & 0 & 0 & 0 \\
0 & 0 & 0 & 0 & 0 & 0 & 0 & 0 & 0 
\end{bmatrix} \\[1ex]
B_6 &= \begin{bmatrix}
0 & 0 & 0 & 0 & 0 & 0 & 0 & 0 & 0 \\
0 & 0 & 0 & 0 & r p_2 & r p_2 & 0 & 0 & 0 \\
0 & 0 & 0 & p_3 & p_3 & p_3 & 0 & 0 & 0 \\
0 & 0 & 0 & 0 & 0 & 0 & 0 & 0 & 0 \\
0 & r (1-\alpha) & r (1-\alpha) & 0 & r (1-\alpha) & r (1-\alpha) & 0 & r (1-\alpha) & r (1-\alpha) \\
-\alpha & -\alpha & -\alpha & -\alpha & -\alpha & -\alpha & -\alpha & -\alpha & -\alpha \\
0 & 0 & 0 & 0 & 0 & 0 & 0 & 0 & 0 \\
0 & 0 & 0 & 0 & 0 & 0 & 0 & 0 & 0 \\
0 & 0 & 0 & 0 & 0 & 0 & 0 & 0 & 0 
\end{bmatrix} \\[1ex]
B_7 &= \begin{bmatrix}
0 & 0 & 0 & 0 & 0 & 0 & 0 & 0 & 0 \\
0 & 0 & 0 & 0 & 0 & 0 & 0 & 0 & 0 \\
0 & 0 & 0 & 0 & 0 & 0 & 0 & 0 & 0 \\
\alpha & \alpha (1-q) & \alpha & \alpha & \alpha (1-q) & \alpha & \alpha & \alpha (1-q) & \alpha \\
0 & 0 & 0 & 0 & 0 & 0 & 0 & 0 & 0 \\
0 & 0 & 0 & 0 & 0 & 0 & 0 & 0 & 0 \\
0 & -q & 0 & 0 & -q & 0 & 0 & -q & 0 \\
0 & 0 & 0 & 0 & 0 & 0 & 0 & 0 & 0 \\
0 & 0 & 0 & 0 & 0 & 0 & 0 & 0 & 0 
\end{bmatrix} \\[1ex]
B_8 &= \begin{bmatrix}
0 & 0 & 0 & 0 & 0 & 0 & 0 & 0 & 0 \\
0 & 0 & 0 & 0 & 0 & 0 & 0 & 0 & 0 \\
0 & 0 & 0 & 0 & 0 & 0 & 0 & 0 & 0 \\
0 & \alpha q & 0 & 0 & \alpha q & 0 & 0 & \alpha q & 0 \\
\alpha & \alpha (1-r) & \alpha (1-r) & \alpha & \alpha (1-r) & \alpha (1-r) & \alpha & \alpha (1-r) & \alpha (1-r) \\
0 & 0 & 0 & 0 & 0 & 0 & 0 & 0 & 0 \\
0 & q & 0 & 0 & q & 0 & 0 & q & 0 \\
0 & -r & -r & 0 & -r & -r & 0 & -r & -r \\
0 & 0 & 0 & 0 & 0 & 0 & 0 & 0 & 0 
\end{bmatrix} \\[1ex]
B_9 &= \begin{bmatrix}
0 & 0 & 0 & 0 & 0 & 0 & 0 & 0 & 0 \\
0 & 0 & 0 & 0 & 0 & 0 & 0 & 0 & 0 \\
0 & 0 & 0 & 0 & 0 & 0 & 0 & 0 & 0 \\
0 & 0 & 0 & 0 & 0 & 0 & 0 & 0 & 0 \\
0 & r \alpha & r \alpha & 0 & r \alpha & r \alpha & 0 & r \alpha & r \alpha \\
\alpha & \alpha & \alpha & \alpha & \alpha & \alpha & \alpha & \alpha & \alpha \\
0 & 0 & 0 & 0 & 0 & 0 & 0 & 0 & 0 \\
0 & r & r & 0 & r & r & 0 & r & r \\
0 & 0 & 0 & 0 & 0 & 0 & 0 & 0 & 0 
\end{bmatrix}
\end{align*}

\end{document}